\documentclass[a4paper,11pt]{article}
\usepackage{hyperref}
\usepackage{graphicx,color,xcolor}
\usepackage{amsthm,amsmath,amssymb,amsfonts}
\usepackage[numbers, sort & compress]{natbib}

\usepackage{hyperref}

\def\pref#1{(\ref{#1})}

\makeatletter
\def\p@enumii{}
\makeatother

\newtheorem{thm}{Theorem}[section]
\newtheorem{lem}[thm]{Lemma}
\newtheorem{Wrong}[thm]{Incorrect Theorem}
\newtheorem{prop}[thm]{Proposition}
\newtheorem{cor}[thm]{Corollary}
\theoremstyle{definition}
\newtheorem{defn}[thm]{Definition}
\newtheorem{ex}[thm]{Example}
\newtheorem{nt}[thm]{Notation}

\theoremstyle{remark}
\newtheorem{rem}[thm]{Remark}
\numberwithin{equation}{section}

\newcommand{\z}{\mathbb{Z}}

\newcommand{\F}{\mathbb{F}}
\newcommand{\chara}{\mathrm{char}\,}

\newcommand{\rg}{\rangle}
\renewcommand{\lg}{\langle}
\newcommand{\se}{\subseteq}
\newcommand{\ann}{\mathrm{ann}}

\renewcommand{\iff}{\Leftrightarrow}
\newcommand{\give}{$\Rightarrow$}
\newcommand{\rgive}{$\Leftarrow$}
\newcommand{\ifof}{if and only if }

\renewcommand{\l}{\left}
\renewcommand{\r}{\right}

\def\f{\frac}
\newcommand{\bc}{\mathbf{c}}

\sloppypar

\begin{document}
\title{One-Sided Repeated-Root Two-Dimensional Cyclic and Constacyclic Codes}
\author{Marziyeh Beygi Khormaei$^1$,  Ashkan Nikseresht$^{2,}$\footnote{corresponding author} \  and Shohreh Namazi$^3$\\
\it\small Department of Mathematics, College of Sciences, Shiraz University, \\
\it\small 71457-13565, Shiraz, Iran\\
\it\small $^1$ E-mail: m64.beygi@gmail.com \\
\it\small $^2$ E-mail: ashkan\_nikseresht@yahoo.com\\
\it\small $^3$ E-mail: namazi@shirazu.ac.ir }
\date{}
\maketitle
%\begin{document}
\begin{abstract}
In this paper, we study some repeated-root two-dimensional cyclic and consta\-cyc\-lic codes over a finite field
$\F=\F_q$. We obtain the generator matrices and generator polynomials of these codes and their duals. We also
investigate when such codes are self-dual. Moreover, we prove that if there exists an asymptotically good family of
one-sided repeated-root two-dimensional cyclic or constacyclic codes, then there exists an asymptotically good family
of simple root two-dimensional cyclic or constacyclic codes with parameters at least as good as the first family.
Furthermore, we show that several of the main results of the papers Rajabi and Khashyarmanesh (2018) \cite{rajabi}
and Sepasdar and Khashyarmanesh (2016) \cite{sepasdar} are not accurate and find other conditions needed for them to
hold.
\end{abstract}
%@@@@@@@@@@@@@@@@@@@@@@@@@@@@@@@@@@@@@@@@@@@@@@@@@@@@@@@@@@@@@@@@@@@@@@@@
\section{Introduction}
Two-dimensional (2D, for short) cyclic codes which have a long history, see for example \cite{ik, im}, still gain
attention, see \cite{Gun,sepasdar,roy, Gun3, Gun quasi} and the references therein. As mentioned in \cite{Gun quasi},
these codes are special cases of quasi-cyclic codes which form an important and well-studied class of codes (see, for
example, \cite{IEEE,Aydin,chen, Sole, nature, Rajput} and their references). Also constacyclic codes which are a
generalization of cyclic codes are investigated over finite fields and some other types of rings, see \cite{beygi,
consta} and their references. In \cite{rajabi}, 2D constacyclic codes were introduced and studied as a generalization
of 2D cyclic codes.

As in one-dimensional cyclic codes, it is more common to consider the simple root case, that is, when  both dimensions
of the code is coprime to the size of the base field, mainly because in this case the structure of the code can be
characterized by the set of its zeros. In the one-dimensional case, it is known that some optimal repeated-root cyclic
codes exist (see \cite{van lint}). Such results caused several authors to study one-dimensional repeated-root cyclic
codes, see \cite{Kumar,Dinh, van lint,ma, casta} and the references therein. In this paper, we consider two-dimensional
cyclic codes in which at least one dimension is coprime to the size of the base field but the other dimension is
arbitrary. We indeed state our results for two-dimensional constacyclic codes which are generalizations of
two-dimensional cyclic codes.

We recall the definition of 2D constacyclic codes (\cite{rajabi}). We always assume that $p$ is a prime number,
$\F=\F_q$ is a finite field with $q=p^r$ elements and $\lambda$ and $\delta$ are units in $\F$. Consider
\begin{eqnarray*}
\tau _{\lambda }:\F^n & \longrightarrow  & \F^n\\
(d_{0},d_{1},\ldots ,d_{n-1}) &\longmapsto & (\lambda
d_{n-1},d_{0},\ldots ,d_{n-2}), \ \ \ \mathrm{where}\ d_j \in \F
\end{eqnarray*}
and
\begin{eqnarray*}
\Upsilon _{\delta }:(\F^n)^m & \longrightarrow  & (\F^n)^m\\
(\mathbf{a }_{0},\mathbf{a }_{1},\ldots ,\mathbf{a }_{m-1})
&\longmapsto & (\delta \mathbf{a }_{m-1},\mathbf{a }_{0},\ldots
,\mathbf{a }_{m-2}), \ \ \ \mathrm{where}\ \mathbf{a }_j \in \F^n.
\end{eqnarray*}
Assume that $\mathbf{a}=(\mathbf{a }_{0},\mathbf{a }_{1},\ldots
,\mathbf{a }_{m-1})$ is an element of $\F^{nm}$, where $\mathbf{a
}_j =(a_{j0},a_{j1},\ldots ,a_{j\ n-1}) \in \F^n$. For any $i,j$,
$0\leq j \leq m-1$ and $0\leq i \leq n-1$, define
 \begin{center}
$\Theta ^{j,i}_{\delta , \lambda}(\mathbf{a})=\Upsilon _{\delta
}^j(\tau _{\lambda }^i(\mathbf{a }_{0}), \tau _{\lambda
}^i(\mathbf{a }_{1}), \ldots, \tau _{\lambda }^i(\mathbf{a
}_{m-1}))$.
\end{center}
%================================================================
A 2D linear code $D$ of length $nm$ is called {\it $( \lambda, \delta)$-consta\-cyc\-lic code} over $\F$, if $\Theta
^{j,i}_{\delta , \lambda}(D)=D$ for any $0\leq j \leq m-1$ and $0\leq i \leq n-1$. In $\F^{nm}\simeq M_{m\times
n}(\F)$, any $nm$-array $(\mathbf{a }_{0},\mathbf{a }_{1},\ldots ,\mathbf{a }_{m-1})$ corresponds to a polynomial in
$\F[x,y]$ with $x$-degree less than $n$ and $y$-degree less than $m$, say $a(x,y)=\sum _{j=0}^{m-1} \sum _{i=0}^{n-1}a
_{ji}x^{i}y^{j}$. With this correspondence, any $( \lambda, \delta)$-consta\-cyc\-lic code of length $nm$ over $\F$ is
identified with an ideal of the quotient ring $\mathcal{S}=\frac{\F[x,y]}{\langle x^{n}-\lambda , y^{m}-\delta \rangle
}$.

In Theorem 4.1 of \cite{sepasdar} which is its main theorem, a set of generators for 2D cyclic codes of length $s2^k$
over fields of odd characteristic is presented. In \cite{rajabi}, this theorem is used to find a generating set for
certain 2D constacyclic codes of length $(2p^s)2^k$ and their duals over fields of odd characteristic. Here in Section
2, we present a counterexample to \cite[Theorem 4.1]{sepasdar} and find the other conditions needed for this theorem to
hold. We show that because of the inaccuracy in this theorem, several results of \cite{rajabi} are also erroneous. Then
in Section 3, we study the algebraic structure of constacyclic codes of length $nm$ where at least one of $n$ or $m$ is
relatively prime to $p$. Moreover, we show how we can find the dimension and a generator matrix of such codes. In
Section 4, we study the dual of such constacyclic codes and present a method to find the parity check matrix of these
codes. We also investigate when such a code is self-dual. Finally in Section 5, we prove that if there exists an
asymptotically good family of such (repeated-root) 2D constacyclic codes, then there exists an asymptotically good
family of simple root 2D constacyclic codes with parameters at least as good as the first family.

 %We shall determine the unique
%generating set of the dual of a $(\lambda,\delta) $-consta\-cyc\-lic code as an ideal of $ \mathcal{T}=\frac{\F[x,y]}{<x^{n}-\lambda^{-1}, y^m-\delta^{-1}>}$.
%\par Also,
% for any repeated-root $(\lambda,\delta) $-consta\-cyc\-lic code $C$, we use the  minimum distance of some  simple-root 2D consta\-cyc\-lic codes to obtain the minimum distance of $C$.
%@@@@@@@@@@@@@@@@@@@@@@@@@@@@@@@@@@@@@@@@@@@@@@@@@@@@@@@@@@@@@@@@@@@@@@@
%%%%%%%%%%%%%%%%%%%%%%%%%%%%%%%%%%%%%%%%%%%%%%%%%%%%%%%%%%%
                                 \section{A counterexample and a correction to some previous results}
%%%%%%%%%%%%%%%%%%%%%%%%%%%%%%%%%%%%%%%%%%%%%%%%%%%%%%%%%%%
First, by presenting a counterexample to the main theorem (Theorem 4.1) of \cite{sepasdar}  we show that this theorem
is not accurate. Consequently, as all of the main theorems of \cite{rajabi} use this incorrect result, they are not
correct either. First we recall the statement of \cite[Theorem 4.1]{sepasdar}. We note that in \cite{sepasdar} the
authors work over a finite field $\F$ with $\chara \F\neq 2$ and consider 2D cyclic codes of length $n=2^ks$ which
correspond to the ideals of the ring $\F[x,y]/\lg x^s-1, y^{2^k} -1 \rg$.

\begin{Wrong}[{\cite[Theorem 4.1]{sepasdar}}]\label{WrongThm}
Suppose that $C$ is a 2D cyclic code of length $n=2^ks$. Then the corresponding ideal $I$ has the
generating set of polynomials as follows:
\begin{align}
    I= \Bigg\lg &  p_1(x)\l(\sum_{i=0}^{2^k-1}y^i\r), p_2(x)\l(\sum_{i=0}^{2^k-1}(-1)^iy^i\r),
     p_3(x)\l(\sum_{i=0}^{2^{k-1}-1}(-1)^iy^{2i}\r) \notag \\
     & p_4(x)\l(\sum_{i=0}^{2^{k-2}-1} (-1)^iy^{4i}\r), \ldots,p_{k+1}(x)\l(\sum_{i=0}^{2^1-1} (-1)^iy^{2^{k-1}i}\r)
     \Bigg\rg. \label{eq1}
\end{align}

\end{Wrong}

The counterexample presented here is the case that $k=2$. Note that in the case $k=2$, this theorem indeed says that
every ideal $I$ of the ring $\F[x,y]$ which contains $x^s-1$ and $y^4-1$ can be generated using $x^s-1, y^4-1$ and
three polynomials of the form $p_1(x)(1+y+y^2+y^3), p_2(x)(1-y+y^2-y^3), p_3(x)(1-y^2)$.

\begin{ex}\label{counterex}
Suppose that $\F=\F_5$ and $I=\lg x^s-1, y+2 \rg$ for an integer $s\geq 1$. Note that $y^4-1=(y+1)(y-1)(y+2)(y-2)$, so
that $y^4-1\in I$. If
 $$I=\lg x^s-1, y^4-1,p_1(x)(1+y+y^2+y^3), p_2(x)(1-y+y^2-y^3), p_3(x)(1-y^2)\rg,$$  then as
$p_3(x)(1-y^2)\in \lg x^s-1, y+2\rg$ and by evaluating at $y=-2$ we have $2p_3(x)\in \lg x^s-1\rg$. Hence
$p_3(x)(1-y^2)$ can be dropped from the aforementioned generating set of $I$. Thus it follows that
 $$y+2\in \lg x^s-1, y^4-1,p_1(x)(1+y+y^2+y^3), p_2(x)(1-y+y^2-y^3)\rg.$$
Now by evaluating the above relation  at $x=1$, we get that $y+2\in \lg y^4-1, 1+y+y^2+y^3, 1-y+y^2-y^3\rg=\lg 1+y^2
\rg$ in $\F[y]$, which is a contradiction. Therefore,  $I$ has not a generating set as in Theorem \ref{WrongThm} and
that theorem is not correct. \qed
\end{ex}

In the proof stated for Theorem  \ref{WrongThm} in \cite{sepasdar}, the problem arises from the fact that $I_3, I_4,
\ldots, I_{k+1}$ (see pages 104, 109 or 111 of \cite{sepasdar}) are not always ideals as assumed in \cite{sepasdar}.
For instance, in the Example \ref{counterex} we have
$$y+2 \in I_3=\{g(x,y)\in \F[x,y]/\lg x^s-1, y^2-1 \rg |g(x,y)(1-y^2)\in I\}$$
but $(1+2y)(1-y^2)\notin  I$ (that is, $1+2y \notin I_3$), although in $\F[x,y]/\lg x^s-1, y^2-1 \rg$, it holds that
$(y+2)y=1+2y$. Next we investigate for which $\F, k, s$ Theorem \ref{WrongThm} does hold.

\begin{prop}\label{correction}
Fix a field $\F$ with an odd characteristic and consider integers $s>1$ and $k>0$. Then every 2D cyclic code of length
$2^{k}s$ over $\F$ as an ideal of $R=\F[x,y]/\lg x^s-1,y^{2^k}-1\rg$  has a generating set of the form \eqref{eq1}, if
and only if both of the following conditions hold:
\begin{enumerate}
\item \label{correction 1} $x^{2^{k-1}}+1$ is irreducible in $\F[x]$;
\item \label{correction 2} $\gcd (x^s-1, x^{2^k}-1) | x^2-1$ in $\F[x]$.
\end{enumerate}
\end{prop}
\begin{proof}
(\give): Note that in \eqref{eq1}, the polynomials $p_i(x)$ can be chosen to be monic factors of $x^s-1$, since the
ring $\F[x]/\lg x^s-1 \rg$ is a principal ideal ring. Assume that in $\F[x]$ the number of monic factors of $x^s-1$ is
$a$. Then the number of ideals of $R$ is at most $a^{k+1}$. On the other hand, if $y^{2^k}-1=f_1\cdots f_t$ is an
irreducible factorization of $y^{2^k}-1$, then by the Chinese Remainder Theorem (CRT)
\begin{equation}
 R\cong \f{ (\F[y]/\lg y^{2^k}-1 \rg)[x]}{\lg x^s-1\rg} \cong \f{\l( \bigoplus_{i=1}^t \F[y]/\lg f_i \rg \r)[x]}{\lg x^s-1
 \rg} \cong \bigoplus_{i=1}^t \f{K_i[x]}{\lg x^s -1 \rg}, \label{eq2}
\end{equation}
where $K_i=\F[y]/\lg f_i \rg$ is an extension field of $\F$. Because in every field we have
\begin{equation}
y^{2^k}-1=(y-1)(y+1)(y^2+1) \cdots (y^{2^{k-1}}+1), \label{eq3}
\end{equation}
it follows that $t\geq k+1$. Suppose $x^s-1$ has $b_i$ monic factors in the extension field $K_i$. Then we have
$b_i\geq a$ for all $i$. Thus the number of ideals of $R$ is $\prod_{i=1}^tb_i \geq a^t\geq a^{k+1}$. Consequently,
$a^{k+1}= \prod_{i=1}^tb_i$ and hence $t=k+1$ and $b_i=a$ for all $i$. From $t=k+1$ it follows that, for each $1\leq
i\leq k-1$ the polynomial $y^{2^i}+1$ is irreducible over $\F$, in particular, $x^{2^{k-1}}+1$ is irreducible in
$\F[x]$. Also we can assume that $f_1=y-1$ and $f_i=y^{2^{i-2}}+1$, for $i\geq 2$. From $b_i=a$ we deduce that the
number of irreducible factors of $x^s-1$ is the same in $\F[x]$ and $K_i[x]$. Now if condition \eqref{correction 2}
does not hold, then an irreducible factor of $x^{2^k}-1$, say $x^{2^i}+1$ for some $1\leq i\leq k-1$, divides $x^s-1$.
Since in $K_{i+2}$ we have $y^{2^i}+1=0$, in $K_{i+2}[x]$ we have $x-y|x^{2^i}+1$ and $x-y$ is a monic factor of
$x^s-1$ over $K_{i+2}[x]$ but not over $\F[x]$. Hence $b_{i+2}>a$, a contradiction from which \eqref{correction 2}
follows.

(\rgive): By \pref{correction 1} it follows that $x^{2^i}+1$ is irreducible over $\F$ for all $i\leq k-1$. So
\eqref{eq3} is an irreducible factorization and we deduce that $t=k+1$. If for some $i\geq 1$, a factor of $x^{2^i}+1$
in $K_i[x]$, divides $x^s-1$, then as $x^{2^i}+1$ is irreducible in $\F[x]$, we must have $x^{2^i}+1|x^s-1$, which
contradicts \pref{correction 2}. Thus the factors of $x^s-1$ over $\F$ and over $K_i$ are the same. So $b_i=a$ for all
$i$ and hence $R$ has $a^{k+1}$ ideals.

Therefore, to show that every ideal is as in \eqref{eq1}, we just need to show that the number of ideals of $R$ which
have the claimed form is $a^{k+1}$. As we can choose all $p_i(x)$ from the set of monic factors of $x^s-1$ in $a^{k+1}$
ways, it suffices to prove that for different choices of $p_i(x)$,  different ideals are generated by Equation
\eqref{eq1}.

To see this, let $f_1=y-1$, $f_i=y^{2^{i-2}}+1$ and $\hat{f}_i=(y^{2^k}-1)/f_i$ and note that
$\hat{f}_1=\sum_{i=0}^{2^k-1}y^i$ and $\hat{f}_{j+1}= -\sum_{i=0}^{2^{k-j}-1} (-1)^iy^{2^j i}$ for all $0\leq j\leq k$.
Thus under the isomorphism \eqref{eq2}, $\hat{f}_i$ is mapped to the vector with all entries zero except for the $i$-th
entry which is a unit in $K_i[x]/\lg x^s-1 \rg$. Thus an ideal $I$ satisfying \eqref{eq1} is mapped to the ideal
generated by the vector with the only nonzero entry at the $i$-th position and equal to $p_i(x)$. Hence if $I$ and $I'$
both satisfy \eqref{eq1} with $p_i(x)$ and $p'_i(x)$, respectively, where $p_i$ and $p'_i$ are monic factors of
$x^s-1$, then $I=I' \iff\ p_i(x)=p'_i(x)$ for each $1\leq i\leq k+1$.
\end{proof}

Using known results on irreducibility and factors of $x^t\pm 1$ over finite fields, we can restate Proposition
\ref{correction} as follows.
\begin{thm}\label{correction2}
Fix a field $\F$ with an odd characteristic and consider integers $s>1$ and $k>0$. Then every 2D cyclic code of length
$2^{k}s$ over $\F$ as an ideal of $R=\F[x,y]/\lg x^s-1,y^{2^k}-1\rg$  has a generating set of the form \eqref{eq1}, if
and only if  one of the following conditions holds:
\begin{enumerate}
\item $k=1$;
\item $k=2$, $q\equiv 3 \mod 4$ and $4\nmid s$.
\end{enumerate}
\end{thm}
\begin{proof}
It is clear that if $k=1$, then both conditions of Proposition \ref{correction} hold. If $k>2$, then it follows from
\cite[Theorem 3.75]{lidl}, that $x^{2^{k-1}}+1$ is not irreducible and Proposition \ref{correction}\pref{correction 1}
does not hold. If $k=2$, then using \cite[Theorem 3.75]{lidl}, we see that condition \pref{correction 1} of
\ref{correction} holds \ifof $q\equiv 3 \mod 4$. Also by \cite[Corollary 3.7]{lidl},
$\gcd(x^s-1,x^{2^k}-1)=x^{\gcd(s,2^k)}-1$ and hence the condition \pref{correction 2} holds \ifof $4\nmid s$.
Consequently, the result follows Proposition \ref{correction}.
\end{proof}

At the end we mention that Theorems 3.2, 4.1 and  5.7 of \cite{rajabi} are not correct either, because they are based
on Theorem \ref{WrongThm}. Indeed, the first part of \cite[Theorem 4.1]{rajabi} (about the generating set of the code)
is just an special case of Theorem \ref{WrongThm} and Example \ref{counterex} serves as a counterexample to that
theorem, too. Also if in Example \ref{counterex}, we replace $x^s-1$ with $x^{2p^s}+1$, then the same argument shows
that the image of the ideal $I$ in the ring $\F[x,y]/\lg x^{2p^s} +1, y^{2^k}-1 \rg$ has no generating set as mentioned
in \cite[Theorem 3.2]{rajabi}. For Theorem 5.7 of \cite{rajabi}, note that it is proved in \cite{rajabi} that the ring
in Theorem 5.7 is isomorphic to a ring as in Theorems 3.2 or 4.1. Thus by this isomorphism, counterexamples to the
latter theorems map to counterexamples of Theorem 5.7. It should be mentioned that Theorem \ref{correction2} or an
argument quite similar to its proof and also the isomorphism defined in \cite[Section 5]{rajabi}, can be applied in
conditions of Theorems 3.2, 4.1 and 5.7 of \cite{rajabi} to see exactly when they are correct. But here in the next
section, we more generally present generating sets for all consta\-cyc\-lic 2D codes which are repeated root in at most
one direction.
%@@@@@@@@@@@@@@@@@@@@@@@@@@@@@@@@@@@@@@@@@@@@@@@@@@@@@@@@@@@@@@@@@@@@@@@@@@@@@@@
\section{One-sided repeated-root 2D consta\-cyc\-lic codes}\label{sec2}
%==============================================

In this paper, we deal with 2D constacyclic codes which are either simple root or have repeated roots in at most one
direction. We call such codes one-sided repeated-root codes, as defined below.
%==========================
\begin{defn}
We call a two-dimensional $(\lambda, \delta)$-consta\-cyc\-lic code $D$ of length $nm$ over $\F_{p^r}$, one-sided
repeated root, if either $\gcd (n,p)=1$ or $\gcd (m,p)=1$.
\end{defn}
First, we fix some notations.
\begin{nt}
From now on, we assume that $n,m$ are two integers, such that $\gcd(n,p)=1$, $m=m'p^s$ and $\gcd(m',p)=1$. Also we
assume that $\lambda, \delta$ are non-zero elements of $\F$. We let $\mathcal{S}=\frac{\F[x,y]}{\langle x^n-\lambda ,
y^m-\delta\rangle}$. Moreover, we assume that $x^n-\lambda =\prod_{j=1}^{\eta}f_j(x)$, where $f_j(x)$, $1\leq j \leq
\eta,$ are monic irreducible coprime polynomials in $\F[x]$. Also we set $d_j=\deg f_j$, $K_j=\frac{\F[x]}{\langle f_j(x) \rangle}\cong \F_{q^{d_j}}$
 and $\mathcal{S}_j=\frac{K_j[y]}{\langle y^m -\delta \rangle}$. We consider elements of $\mathcal{S}$ as
those elements of $\F[x,y]$ whose $x$-degree and $y$-degree is less than $n$ and $m$, respectively. A similar notation
holds for elements of $\mathcal{S}_j$ and $K_j$.
\end{nt}
Note that every $(\lambda, \delta)$-consta\-cyc\-lic code of length $nm$ over $\F$ is an ideal of $\mathcal{S}$. As the
following remark shows, $\mathcal{S}$ is the direct sum of $\mathcal{S}_j$'s. By using this structure, we determine the
ideals of $\mathcal{S}$.
%==============================================================================
\begin{rem}\label{r1}
Using the CRT we see that
\begin{align*}
\mathcal{S}& = \frac{\F[x,y]}{\langle x^n-\lambda , y^m-\delta\rangle}\cong \frac{\frac{\F[x]}{\langle x^n-\lambda
\rangle}[y]}{\langle y^m -\delta \rangle}\cong \frac{  \left( \bigoplus_{j=1}^{\eta}\frac{\F[x]}{\langle
f_j(x)\rangle}\right)[y]}{\langle y^m -\delta \rangle} \\
 & \cong \bigoplus_{j=1}^{\eta}\frac{\frac{\F[x]}{\langle
f_j(x)\rangle}[y]}{\langle y^m -\delta \rangle}=\bigoplus_{j=1}^{\eta} \mathcal{S}_j.
\end{align*}
This isomorphism is $\psi:\mathcal{S} \to \bigoplus_{j=1}^\eta \mathcal{S}_j$ with
 $$\psi(h(x,y))=(\psi_1(h(x,y)), \ldots,  \psi_\eta(h(x,y))),$$
where $\psi_j:\mathcal{S} \to \mathcal{S}_j$ is defined with $\psi_j(h(x,y))=  h(x,y)\ \mod f_j(x)$.
\end{rem}

%=================================================================================

%========================================================================
Now, we can determine the general form of ideals of $\mathcal{S}$. %Note that $K_j$ is a field with $q^{d_j}$ elements.
If $C$ is a $(\lambda , \delta)$-consta\-cyc\-lic code over $\F$, then $C$ is an ideal of $\mathcal{S}$. Hence $\psi
(C)=\bigoplus\limits _{j=1}^{\eta}C_j$, where $C_j$, $1\leq j\leq \eta$, is an ideal of
$\mathcal{S}_j=\frac{K_j[y]}{\langle y^m -\delta \rangle}.$ In fact, $C_j$ is a $\delta$-consta\-cyc\-lic code over $K_j$. Thus as an ideal of $\mathcal{S}_j$,
$C_j=\ll g_j(x,y)\gg$, where $g_j(x,y)\in K_j[y] $ is the generator polynomial for $C_j$, %in $K_j[y] $,
that is, the unique monic
polynomial of minimum $y$-degree in $C_j$ which divides $y^m-\delta$ in $K_j[y] $. In particular, $C_j=0$ \ifof
$g_j(x,y)=y^m-\delta$.

%====================================================================================
 \begin{thm}\label{p2}
Let $C$ be a $(\lambda , \delta)$-consta\-cyc\-lic code over $\F$. Then there exist unique polynomials $g_j(x,y)$ such
that $g_j(x,y)\mid y^m-\delta$ in $K_j[y]$, $g_j(x,y)$ is monic when considered as a polynomial in $y$ and as an ideal
of $\mathcal{S}$,
  \begin{center}
   $ C=\langle g_1(x,y)\prod \limits _{i\neq 1} f_i(x),g_2(x,y)\prod \limits _{i\neq 2} f_i(x),\ldots ,g_{\eta}(x,y)\prod \limits _{i\neq \eta} f_i(x) \rangle$.
  \end{center}
Moreover, $\dim ( C) =mn-\sum^{\eta}_{j=1}d_jt_j$, where $t_j= \deg_y g_j$.

% $$t_j=\left \{
%  \begin{array}{cc}
%   \deg_y g_j & g_j\neq 0 \\
%   m & g_j=0 \\
%  \end{array}\right. .
%$$
\end{thm}
\begin{proof}
Suppose that $\psi (C)=\bigoplus\limits _{j=1}^{\eta}\ll g_j(x,y)\gg$. Set
\begin{center}
$D=\langle g_1(x,y)\prod \limits _{i\neq 1} f_i(x),g_2(x,y)\prod \limits _{i\neq 2} f_i(x),\ldots ,g_{\eta}(x,y)\prod
\limits _{i\neq \eta} f_i(x) \rangle$.
\end{center}
For any $l\neq j$, we have $\psi _{l}(g_j(x,y)\prod  _{i\neq j} f_i(x))=0$ (because $\psi _{l}( f_l(x))=0$). Also,
$\psi _{l}(g_l(x,y)\prod  _{i\neq l} f_i(x))=g_l(x,y)\prod  _{i\neq l} f_i(x) \mod f_l$ in $S_l$. Since $\prod_{i\neq
l} f_i(x)$ is a unit in $K_l[y]$, so $\psi (D)=\bigoplus_{j=1}^{\eta}\ll g_j(x,y)\gg=\psi (C)$. Thus $C=D$.

Now, let  $C=\langle g'_1(x,y)\prod  _{i\neq 1} f_i(x),\ldots ,g'_{\eta}(x,y)\prod
 _{i\neq \eta} f_i(x) \rangle$, for some $g'_j(x,y)$ such that $g'_j(x,y)\mid y^m-\delta$ in $K_j[y]$ and
$g'_j(x,y)$ is monic when considered as a polynomial in $y$. We claim that $g'_j(x,y)=g_j(x,y)$. For any $j$, we have
$\psi_j(C)=\langle g'_j(x,y)\rangle=\langle g_j(x,y)\rangle$. Since $g_j$ and $g'_j$ are monic and divide $y^m-\delta$,
so $g'_j(x,y)=g_j(x,y)$.

% $C=\langle h_1(x,y),h_2(x,y),\ldots ,h_{\eta}(x,y)\rangle$, where $\deg_y h_j<m$, then $\psi _{l}(h_j(x,y))=0$ for $j\neq l$ and $\ll \psi _{l}(h_l(x,y)) \gg=\ll g_l(x,y)\gg$. We claim that $h_j(x,y)=g_j(x,y)\prod  _{i\neq j} f_i(x)$. Let $h_j(x,y)=\sum  ^{m-1}_{r=0}h_{rj}(x)y^r$, where $h_{rj}(x)\in \F[x]$ and $\deg h_{rj}<n$. We have $\psi _{l}(h_j(x,y))= \sum  ^{m-1}_{r=0}\pi _l(\phi(h_{rj}(x)))y^r$ in $K_l[y]$. So for any $j\neq l$, $\pi _l(\phi(h_{rj}(x)))=0$ in $K_l$. Hence $h_{rj}(x)+\langle f_l(x)\rangle=0$. So in $\F[x]$, $f_l(x)\mid h_{rj}(x)$. Thus $\prod  _{i\neq j} f_i(x)\mid h_{rj}(x)$. Hence $h_j(x,y)=k_j(x,y)\prod  _{i\neq j} f_i(x)$ for some $k_j(x,y)\in \F[x,y]$, with $\deg_x k_j<\deg f_j$. Since $\prod  _{i\neq j} f_i(x)$ is a unit in $K_j$, $\ll (g_j(x,y)) \gg=\ll h_j(x,y)\gg=\ll k_j(x,y)\gg$. So $k_j(x,y)=g_j(x,y)u$ for some unit $u$ in $\mathcal{S} $. \\
%We have $\dim \ C=\dim\ \psi(C)=\sum  ^{\eta}_{j=1}\dim\ C_j$. Thus
Let $C_j=\langle g_j(x,y)\rangle$. Thus
 \begin{eqnarray*}
\mid C\mid &=&\mid \psi(C)\mid
=\mid C_1\mid \mid C_2\mid \cdots \mid C_{\eta}\mid
=(q^{d_1})^{m-t_1}(q^{d_2})^{m-t_2} \cdots (q^{d_{\eta}})^{m-t_{\eta}}\\
&=& q^{m(\sum \limits ^{\eta}_{j=1}d_j)-(\sum \limits ^{\eta}_{j=1}d_j t_j)}
= q^{mn-\sum \limits ^{\eta}_{j=1}d_jt_j}.
\end{eqnarray*}
Consequently, $\dim ( C) =mn-\sum \limits ^{\eta}_{j=1}d_jt_j$.
\end{proof}
%========================================
The above theorem, introduces a generating set for an ideal $C$ of $\mathcal{S}$ %a $(\lambda , \delta)$-consta\-cyc\-lic codes over $\F$
and shows
that this generating set is unique. In the sequel,  by
\begin{center}
$ C=\ll g_1(x,y)\prod \limits _{i\neq 1} f_i(x),\ldots ,g_{\eta}(x,y)\prod \limits _{i\neq \eta} f_i(x) \gg$.
\end{center}
we mean that all $g_j(x,y)$ satisfy the conditions of Theorem \ref{p2} and
\begin{center}
$ C=\lg g_1(x,y)\prod \limits _{i\neq 1} f_i(x),\ldots,g_{\eta}(x,y)\prod \limits _{i\neq \eta} f_i(x) \rg$.
\end{center}
%===========================================================
\begin{rem}
With the notations of Theorem \ref{p2}, the polynomial $\sum _{j=1}^{\eta}g_j(x,y)\prod_{i\neq j} f_i(x)$ generates
$C$. To see this, assume that $D=\langle \sum_{j=1}^{\eta}g_j(x,y)\prod_{i\neq j} f_i(x)\rangle$. Note that each
$\psi_l$ is a ring isomorphisms and $\psi_l(g_j(x,y)\prod_{i\neq j}f_i(x))=0$, if $l\neq j$. It follows that
$\psi_l(C)= \psi_l(D)$ for all $l$ and hence $C=D$. Note that $\mathcal{S}$ is a principal ideal ring (PIR).
\end{rem}
%==============================================================
Let $C=\langle h_1(x,y),h_2(x,y),\ldots ,h_{r}(x,y)\rangle$, as an ideal of $\mathcal{S}$, where $\deg_y h_j<m$. We
show how to find the unique generator polynomials for $C$ as in the previous theorem. For this, we compute  $\psi
_j(h_l(x,y))$, for  each $j,l$, $1\leq l \leq r, 1\leq j \leq \eta$. Then using Euclidean algorithm, we calculate the
greatest common divisor of $\psi_j(h_l(x,y))$ for all $l$. Assume that
\begin{center}
$g_j(x,y)=\gcd(\psi _j(h_1(x,y)),\ldots , \psi _j(h_{r}(x,y)), y^m-\delta)$ in $K_j[y]$.
\end{center}
Here we assume that $\gcd$ returns a monic polynomial and also $\gcd(0,f)=f$. Thus $g_j(x,y)\mid y^m-\delta $ and is
monic as a polynomial in $y$ and also $\psi_j(C)=\lg g_j(x,y) \rg $. So
$$ C=\ll g_1(x,y)\prod_{i\neq 1}
f_i(x),g_2(x,y)\prod_{i\neq 2} f_i(x),\ldots ,g_{\eta}(x,y)\prod_{i\neq \eta} f_i(x) \gg.$$
%==================================================================
\begin{ex}\label{e0}
 Suppose that $\mathcal{S}=\frac{\F_2[x,y]}{\langle x^3-1 , y^3-1\rangle}$. We have
 \begin{eqnarray*}
x^3-1 &=&(x+1)(x^2+x+1)\ \ \ \ \ \mathrm{in}\ \F_2[x]\\
y^3-1&=&(y+1)(y^2+y+1)\ \ \ \ \ \mathrm{in}\ \frac{\F_2[x]}{\langle x+1\rangle}[y]\\
y^3-1&=&(y+1)(y+x)(y+x+1)\ \ \ \ \ \mathrm{in}\ \frac{\F_2[x]}{\langle x^2+x+1\rangle}[y].
\end{eqnarray*}
If $C$ ia an ideal of $\mathcal{S}$, then $C=\ll g_1(x,y)(x^2+x+1), g_2(x,y)(x+1)\gg $,  where $g_1(x,y)\mid y^3-1$ in
$\frac{\F_2[x]}{\langle x+1\rangle}[y]$ and $g_2(x,y)\mid y^3-1$ in $\frac{\F_2[x]}{\langle x^2+x+1\rangle}[y]$. The
number of such codes is $2^2\times 2^3=32$. Consider the code $C=\langle h_1(x,y), h_2(x,y),h_3(x,y)\rangle$, where
\begin{eqnarray*}
  h_1(x,y) &=& y^2 x^2+y^2x+x+1 ,\\
  h_2(x,y) &=&  y^2 x+yx+x \ \ \mathrm{and}\\
  h_3(x,y) &=& y^2 x^2+y^2+yx^2+y+x^2+1.\\
\end{eqnarray*}
We have
\begin{eqnarray*}
 \psi_1( h_1) &=& 0, \ \  \psi_2( h_1) = y^2 x^2+y^2x+x^2=x(y+x)(y^2+y+x),\\
 \psi_1( h_2) &=& y^2+y+1, \ \  \psi_2( h_2) = y^2 x+yx+x=(y+x)(y+x+1),\\
  \psi_1( h_3) &=& 0 \ \  \  \mathrm{and}\ \ \psi_2( h_3) = y^2 x+yx+x=x(y+x)(y+x+1).\\
\end{eqnarray*}
So
\begin{eqnarray*}
 \gcd(\psi_1( h_1), \psi_1(h_2),\psi_1(h_3),y^3-1) &=& y^2+y+1 \ \mathrm{in}\  \frac{\F_2[x]}{\langle x+1\rangle}[y]\\
 \gcd( \psi_2(h_1), \psi_2(h_2),\psi_2(h_3), y^3-1) &=& y+x \ \mathrm{in}\  \frac{\F_2[x]}{\langle x^2+x+1\rangle}[y].\\
\end{eqnarray*}
Therefore, $C= \ll (1+y+y^2)(x^2+x+1), (y+x)(x+1)\gg$. Note that here $f_1=x+1$, $f_2=x^2+x+1$, $g_1=1+y+y^2$ and
$g_2=y+x$ and hence we get $\dim C=9-4=5$.
\end{ex}
%======================================================
\begin{cor}\label{p3}
   Consider a $(\lambda,\delta)$-consta\-cyc\-lic code $C$ over $\F$ such as
   \begin{center}
   $C=\ll g_1(x,y)\prod \limits _{i\neq 1} f_i(x),g_2(x,y)\prod \limits _{i\neq 2} f_i(x),\ldots ,g_{\eta}(x,y)\prod \limits _{i\neq \eta} f_i(x) \gg$.
    \end{center}
  Then the following set is a basis for $C$ over $\F$.
  \begin{center}
  $\Delta=\bigcup \limits _{j=1}^{\eta} \{x^ry^l g_j(x,y)\prod \limits _{i\neq j} f_i(x)\mid 0 \leq r < d_j,0\leq l<m-t_j \}$,
    %$\Delta=\bigcup \limits _{l=0}^{\eta}(\bigcup \limits _{l=0}^{m-t_l-1}\{y^l g_j(x,y)\prod \limits _{i\neq j} f_i(x),y^l x g_j(x,y)\prod \limits _{i\neq j} f_i(x),\ldots ,y^l x^{d_j-1}g_{j}(x,y)\prod \limits _{i\neq j} f_i(x)\}$,
  \end{center}
where $t_j=\deg_y g_j$. Note that, if $g_j=y^m-\delta$, then the set $\{x^ry^l g_j(x,y)\prod_{i\neq j} f_i(x)\mid 0
\leq l < m-t_j \}$ is empty for all $r$.
\end{cor}
\begin{proof}
  Let $D$ be the $\F$-subspase of $\mathcal{S}$, generated by $\Delta$. We show that $D=C$. Since $C$ is an ideal of $\mathcal{S}$, $\Delta\subseteq C$. So $D\subseteq C$. We will show that the elements of $\Delta $ are linearly independent %in $\mathcal{S}$
  over $\F$. Suppose that
  \begin{eqnarray*}
  B=\sum ^{\eta}_{j=1}\left(  \left(\sum \limits ^{m-t_j-1}_{r=0}a_{r j}(x)y^r \right)g_j(x,y)\prod \limits _{i\neq j} f_i(x)  \right)=0
%B &=&\left(\sum \limits ^{m-t_1-1}_{r=0}a_{r 1}(x)y^r \right)g_1(x,y)\prod \limits _{i\neq 1} f_i(x)+ \left(\sum \limits ^{m-t_2-1}_{r=0}a_{r 2}(x)y^r \right)g_2(x,y)\prod \limits _{i\neq 2} f_i(x) \\
%&+&\ldots +
% (\sum \limits ^{m-t_{\eta}-1}_{r=0}a_{r \eta}(x)y^r)g_{\eta}(x,y)\prod \limits _{i\neq \eta} f_i(x) =0
\end{eqnarray*}
in $\mathcal{S}$, where $\deg a_{rj}(x)<d_j$. Hence $\psi(B)=0$. So for any $j$, $1\leq j \leq \eta$, $\psi _j(B)\in
\lg y^m-\delta\rg$. Thus $(\sum  ^{m-t_j-1}_{r=0}a_{r j}(x)y^rg_j(x,y))\psi _j(\prod  _{i\neq j} f_i(x))=0$ in
$\mathcal{S}_j$. Since $\prod  _{i\neq j} f_i(x)$ is a unit in $K_j$, $\sum  ^{m-t_j-1}_{r=0}a_{r j}(x)y^rg_j(x,y)=0$
in $\mathcal{S}_j$. Since the $y$-degree of $\sum  ^{m-t_j-1}_{r=0}a_{r j}(x)y^rg_j(x,y)$ is less than $m$, $\sum
 ^{m-t_j-1}_{r=0}a_{r j}(x)y^rg_j(x,y)=0$ in $K_j[y]$. Thus $(\sum  ^{m-t_j-1}_{r=0}a_{r j}(x)y^r)=0$ in
$K_j[y]$. Hence $a_{r j}(x)=0$ in $K_j$ for all $r$. Since $\deg a_{rj}<\deg f_j,$ $a_{r j}(x)=0$ in $\F[x]$. So all
coefficients of $a_{r j}(x)$ are zero and hence $\Delta$ is an independent set. Now,
 \begin{eqnarray*}
\mid D\mid &=&(q^{d_1})^{m-t_1}(q^{d_2})^{m-t_2}\cdots (q^{d_{\eta}})^{m-t_{\eta}} =q^{mn-\sum \limits ^{\eta}_{j=1}d_jt_j}=
 \mid C\mid.
\end{eqnarray*}
Therefore, $C=D$.
\end{proof}
%\par We obtain a generator matrix for $C$ rows of which are corresponding codewords of elements of $\Delta$.
\par The following example shows how we can find a generator matrix for a given code using Corollary \ref{p3}.
%============================
\begin{ex}\label{ex1}
With the notations of Example \ref{e0}, consider $C= \ll (1+y+y^2)(x^2+x+1), (y+x)(x+1)\gg$. We have $\mid
C\mid=2^5=32$. Also, the generator matrix of $C$ is of the following form.
\begin{equation*}
%\mathbf{\mathbf
{G}= \left[
\begin{array}{c}
(1+y+y^2)(x^2+x+1)\\
(y+x)(x+1)\\
x(y+x)(x+1)\\
y(y+x)(x+1)\\
xy(y+x)(x+1)\\
\end{array}
 \right]   =\left[
\begin{array}{ccccccccc}
1& 1& 1 & 1 & 1& 1& 1& 1 &1\\
0& 1& 1 & 1 & 1& 0& 0& 0 &0\\
1& 0& 1 & 0 & 1& 1& 0& 0 &0\\
0& 0& 0& 0& 1& 1 & 1 & 1& 0\\
0& 0& 0& 1& 0& 1 & 0 & 1& 1
\\
\end{array} \right].
\end{equation*}
Note that every row of $G$ as a constacyclic codeword is seen as a $3\times 3$ array. For example, the array form of
the second row is
 $$\begin{matrix}
 \phantom{1}\\
 1\\
 y\\
 y^2
 \end{matrix}
 \begin{array}{c}
\begin{matrix} 1 &x &x^2 \end{matrix} \\
\begin{bmatrix}
0& 1& 1 \\
 1 & 1& 0\\
  0& 0 &0\\
\end{bmatrix} \end{array} $$
and its cyclic shift corresponding to multiplication by $x$ is the third row.
\end{ex}
%@@@@@@@@@@@@@@@@@@@@@@@@@@@@@@@@@@@@@@@@@@@@@@@@@@@@@@@@@@@@@
\section{The dual of one-sided repeated-root 2D consta\-cyc\-lic codes}\label{sec3}
%@@@@@@@@@@@@@@@@@@@@@@@@@@@@@@@@@@@@@@@@@@@@@@@@@@@@@@@@@@@@@

For any $(\lambda,\delta)$-consta\-cyc\-lic code $C\subseteq \F^{nm}$, let
\begin{center}
 $C^{\perp}= \{\mathbf{u}\in \F^{nm} \mid \mathbf{u}\mathbf{\cdot}\mathbf{w}=0 \    \mathrm{for}\  \mathrm{any} \    \mathbf{w}\in C \}$
 \end{center}
be the dual of the code $C$, where $\mathbf{u}\mathbf{\cdot }\mathbf{w}$ is the Euclidean inner product of $\mathbf{u}$
and $\mathbf{w}$ in $\F^{nm}$. By \cite[Propositin 2.2]{rajabi}, $C^\perp$  is a $(\lambda^{-1},\delta^{-1})
$-consta\-cyc\-lic code over $\F$. In this section, we shall determine the unique generating set of the dual of $C$ as
an ideal of $ \mathcal{T}=\frac{\F[x,y]}{<x^{n}-\lambda^{-1}, y^m-\delta^{-1}>}$.
%=============================================
\par If $f(x)$ is a non-zero polynomial of degree $d$ in $\F[x]$, we define
%the reciprocal of $f(x)$ by
$f^{\ast}(x)=x^df(x^{-1})$. Since $x^n-\lambda =\prod  _{i= 1}^{\eta} f_i(x)$, we have $x^n-\lambda^{-1} =u\prod
 _{i= 1}^{\eta} f^{\ast}_i(x)$ for some $u\in \F$.
%===============================================================
\par Similar to Remark \ref{r1}, we can show that
$$ \mathcal{T}\simeq \bigoplus_{i=1}^{\eta}\frac{\frac{\F[x]}{\langle f^{\sharp}_i(x)\rangle}[y]}{\langle y^m -\delta ^{-1} \rangle}, \quad
\text{where}\quad f^{\sharp}_i(x)=\frac{f^{\ast}_i}{f_i(0)}.$$
%==========3.1==========================================================================
\begin{prop}[{\cite[Proposition 2.1]{rajabi}}]\label{kh}
Assume that $\lambda _1$ and $\lambda _2$ are non-zero elements of $\F$ and $f(x,y)=f_0(x)+\cdots +f_{m-1}(x)y^{m-1},\
g(x,y)=g_0(x)+\cdots +g_{m-1}(x)y^{m-1}\in \F[x,y]$, where $f_i(x)=\sum  ^{n-1}_{j=0}a_{ji}(x)x^j$ and $g_i(x)=\sum
 ^{n-1}_{j=0}b_{ji}(x)x^j$ for $i=0,1,\ldots ,m-1$. Then $f(x,y)g(x,y)=0$ in $\frac{\F[x,y]}{<x^{n}-\lambda
_{1}, y^m-\lambda_2>}$ iff $(a_0, a_1,\ldots ,a_{m-1})$is orthogonal to $(b_{m-1},\ldots ,b_0)$ and all its
$(\lambda_1^{-1},\lambda_2^{-1})$-consta\-cyc\-lic shifts, where $a_i=(a_{0i},\ldots a_{n-1,i})$,
$b_i=(a_{n-1,i},\ldots a_{0i})$ for $0\leq i\leq m-1$.
\end{prop}
%=================================================================
Suppose that $f(x,y)\in \F[x,y]$ has $x$- and $y$-degree less than $n$ and $m$, respectively. Then we define
$f^\circledast(x,y)=x^{n-1}y^{m-1}f(\frac{1}{x},\frac{1}{y})$. If $f\in \mathcal{S}$ (resp. in $\mathcal{T}$), we consider
$f^\circledast$ as an element of $\mathcal{T}$ (resp. $\mathcal{S}$). With this definition, Proposition \ref{kh} says that
$\ann(g)=\lg g^\circledast \rg ^\bot$ for any $g\in \mathcal{S}$. It is easy to see that $f^{\circledast\circledast}(x,y)=f(x,y)$ and
if $f,g\in \mathcal{S}$, then and $(f+g)^\circledast=f^\circledast+g^\circledast$.
Now for any polynomial $f(x,y)\in \F[x,y]$ define $f^{\ast}(x,y)=x^{\deg_x f}y^{\deg_y f}f(\frac{1}{x}, \frac{1}{y})$.
Again, if $f\in \mathcal{S}$ (resp. in $\mathcal{T}$), we consider $f^\ast$ as an element of $\mathcal{T}$ (resp. $\mathcal{S}$).
Thus $f^\circledast(x,y)=uf^{\ast}(x,y)$ for some $u\in \mathcal{T}$ and
\begin{equation}\label{e8}
  \langle f^\circledast(x,y)\rangle =\langle f^{\ast}(x,y)\rangle.
\end{equation}
Note that for $f,g\in \mathcal{S}$ by $(fg)^\circledast$ we mean $h^\circledast$ where $h$ is a polynomial with $x$-degree
less than $n$ and $y$-degree less than $m$ such that we have $h=fg$ in $\mathcal{S}$.
Also for $A\se \mathcal{S}$, by $A^\circledast$ we mean $\{a^\circledast| a\in A\}$. The notions $A^\ast$ and $(fg)^\ast$
have similar meanings.
\begin{lem}  \label{fg*}
Suppose that $f,g\in S$ and let $I$ be an ideal of $\mathcal{S}$. Then
\begin{enumerate}
\item \label{fg*1} $(fg)^\circledast =uf^\circledast g^\circledast$ for some unit $u\in \mathcal{T}$;
\item \label{fg*2} $(fg)^\ast =uf^\ast g^\ast$ for some unit $u\in \mathcal{T}$.
\item \label{fg*3} $I^\circledast=\lg I^\ast \rg$.
\end{enumerate}
\end{lem}
\begin{proof}
\pref{fg*1}: Suppose that $fg=h$ in $\mathcal{S}$. Then in $\F[x,y]$ there exist polynomials $a$ and $b$ such that
$fg=h+a(x^n-\lambda)+b(y^m-\delta)$. Since $x$- and $y$-degrees of $fg$ are at most $2n-2$ and $2m-2$, respectively, we can
assume that $\deg_x(a)\leq n-2$, $\deg_y(a)\leq 2m-2$, $\deg_x(b)\leq 2n-2$ and $\deg_y(b)\leq m-2$.
Now in $\F[x,y]$ we have
\begin{eqnarray*}
f^\circledast g^\circledast &=&x^{2n-2}y^{2m-2}f(\f{1}{x},\f{1}{y}) g(\f{1}{x},\f{1}{y})\\
 &=& x^{2n-2}y^{2m-2}h(\f{1}{x},\f{1}{y})+ x^{n-2}y^{2m-2}a(\f{1}{x},\f{1}{y}) (1-\lambda x^n) \\
 & &+x^{2n-2}y^{m-2} b(\f{1}{x},\f{1}{y}) (1-\delta y^m) \\
 &=&x^{n-1}y^{m-1} h^\circledast -\lambda a'(x,y)(x^n-\lambda^{-1}) -\delta b'(x,y) (y^m-\delta^{-1}),
\end{eqnarray*}
for some  $a',b'\in \F[x,y]$. Therefore in $\mathcal{T}$ we have $$f^\circledast g^\circledast=x^{n-1}y^{m-1}
h^\circledast = u(fg)^\circledast,$$ for some unit $u\in \mathcal{T}$.

\pref{fg*2}: This follows from part \pref{fg*1} and \eqref{e8}.

\pref{fg*3}:
It is easy to see that if $d=\deg_x(f), d'=\deg_y(f)$, then $f^\ast=(x^{n-d-1}y^{m-d'-1}f)^\circledast$ and hence $I^\ast\se I^\circledast$.
Also from \eqref{e8}, it follows that $I^\circledast \se \lg I^\ast \rg$. Thus it suffices to show that $I^\circledast$
 is an ideal. It is closed under addition because $(f+g)^\circledast=f^\circledast + g^\circledast$. Now assume that
 $f^\circledast\in I^\ast$, that is $f\in I$, and $g\in \mathcal{T}$.
By part \pref{fg*1}, we have $(gf^\circledast)^\circledast= ug^\circledast f\in I$ for some unit $u\in \mathcal{S}$ and hence $gf^\circledast=
 (gf^\circledast)^{\circledast\circledast}\in I^\circledast$.
\end{proof}
%================================
Using Proposition \ref{kh} and Lemma \ref{fg*}, we have the following lemma that shows the relationship between the elements of $\ann(C)$ and $C^{\perp}$.
%===================================
\begin{lem}\label{l1}
Let $C$ be a $(\lambda , \delta)$-consta\-cyc\-lic code over $\F$. Then
\begin{center}
  $C^{\perp}=\ann ^\circledast(C)=\lg \ann^{\ast}(C) \rg$.
\end{center}
%Let $C$ be a $(\lambda , \delta)$-consta\-cyc\-lic code over $\F$, $g(x,y)\in \mathcal{S}$ and for any $f(x,y)\in C$, $f(x,y)g(x,y)=0$. Then $g^\circledast(x,y)\in C^{\perp}.$
\end{lem}
\begin{proof}
By Lemma \ref{fg*} $\ann ^\circledast(C)=\lg \ann^{\ast}(C)\rg$. Also as $(fg)^{\circledast}=0$ in $\mathcal{T}$ \ifof $fg=0$ in $\mathcal{S}$, it follows that
$\ann^\circledast(C)=\ann(C^\circledast)$. Note that Proposition \ref{kh}, indeed says that $\ann(g)=\lg g^\circledast \rg^\perp$. Now from
 $$C^\perp= \bigcap_{g\in C} \lg g \rg^\perp = \bigcap_{g\in C} \ann(g^\circledast)=\ann(C^\circledast),$$
the claim follows.
\end{proof}
%==========================================

Let $C=\ll g_1(x,y)\prod  _{i\neq 1} f_i(x),\ldots ,g_{\eta}(x,y)\prod \limits _{i\neq \eta} f_i(x) \gg$ be a $(\lambda
, \delta)$-consta\-cyc\-lic code over $\F$. Assume that
%$h_i(x,y)$, $1\leq i \leq \eta$, is a polynomial in $\F[x,y]$ with $\deg_y h_i <m$ and $\deg_x k_i <n$ such that
\begin{equation}\label{e10}
h_i(x,y)=\frac{y^m-\delta}{g_i(x,y)}
\end{equation}
in $K_i[y]$. If $g_i(x,y)=0$, we assume that $h_i(x,y)=1$. Suppose that $h_i^{\sharp}(x,y)$ is the monic polynomial in
$\F[x,y]$ such that
\begin{equation}\label{e99}
 h_i^{\sharp}(x,y))=\frac{h_i^{\ast}(x,y)}{h_i(x,0)}
\end{equation}
in $\frac{\F[x]}{\langle f^{\sharp}_i(x)\rangle}[y]$. With this notations, we have the following theorem that gives the
generating set of the dual of the code $C$.
%=========================================================
\begin{thm}\label{p4}
Let $C=\ll g_1(x,y)\prod  _{i\neq 1} f_i(x),\ldots ,g_{\eta}(x,y)\prod  _{i\neq \eta} f_i(x) \gg$ be a $(\lambda ,
\delta)$-consta\-cyc\-lic code over $\F$. Then
\begin{center}
$C^{\perp}=\ll h_1^{\sharp}(x,y)\prod \limits _{i\neq 1} f^{\sharp}_i(x),h_2^{\sharp}(x,y)\prod \limits _{i\neq 2}
f^{\sharp}_i(x),\ldots ,h_{\eta}^{\sharp}(x,y)\prod \limits _{i\neq \eta} f^{\sharp}_i(x) \gg$.
\end{center}
\end{thm}
\begin{proof}
We can see that $h_j^{\sharp}(x,y)$, $0\leq j\leq \eta$, is a monic polynomial, by \eqref{e99}, and divides
$y^m-\delta^{-1}$ in $\frac{\F[x]}{\lg f_j^\sharp(x) \rg}[y]$, by \eqref{e10}. Let $D=\langle h_1^{\sharp}(x,y)\prod
_{i\neq 1} f^{\sharp}_i(x),\ldots ,h_{\eta}^{\sharp}(x,y)\prod  _{i\neq \eta} f^{\sharp}_i(x) \rangle$. For $l\neq j$,
we have $a(x,y)=\big(g_j(x,y)\prod  _{i\neq j} f_i(x)\big)\big(h_l(x,y)\prod  _{i\neq l} f_i(x)\big)=0$ in
$\mathcal{S}$, because, $x^n-\lambda \mid a(x,y)$ in $\F[x,y]$. For $l=j$, we have $h_l(x,y)g_l(x,y)=0$. So
$h_l(x,y)\prod _{i\neq l} f_i(x)\in \ann(C)$. Hence by Lemma \ref{l1}, $(h_l(x,y)\prod  _{i\neq l} f_i(x))^*\in
C^{\perp}$ for any $l$, $1\leq l\leq \eta$. So by Lemma \ref{fg*}\pref{fg*2},
$u h_l^{\ast}(x,y)\prod  _{i\neq l} f_i^{\ast}(x)\in C^{\perp}$ for some unit $u\in \mathcal{T}$ and hence
$v h_l^{\sharp}(x,y)\prod  _{i\neq l} f_i^{\sharp}(x)\in C^{\perp}$ for
some unit $v\in \mathcal{T}$. Consequently, $D\subseteq C^{\perp}$. Because $h(x,0)\neq 0\neq f_i(0)$, we have $\deg_y
h_i^\sharp=\deg_y h_i^*=\deg_y h_i$ and $\deg f_i^\sharp= \deg f_i$. So
  \begin{eqnarray*}
  \dim (D)&=& mn-\sum \limits ^{\eta}_{j=1}d_j(m-t_j)
  =\sum \limits ^{\eta}_{j=1}d_jt_j=\dim(C^{\perp}).
%\mid D\mid &=&\mid \bigoplus \limits _{j=1}^{\eta}\ll h_j^*(x,y)\gg\mid \\
%&=&\mid \ll h_1^*(x,y)\gg\mid \mid \ll h_2^*(x,y)\gg\mid \ldots \mid \ll h_{\eta}^*(x,y)\gg\mid \\
%&=&(q^{d_1})^{m-(m-t_1)}(q^{d_2})^{m-(m-t_2)}\ldots (q^{d_{\eta}})^{m-(m-t_{\eta})}\\
%&=& q^{\sum \limits ^{\eta}_{j=1}d_j t-j}\\
%&=& q^{\sum \limits ^{\eta}_{j=1}d_jt_j}.
\end{eqnarray*}
%Moreover, $\mid C^{\perp} \mid=\frac{q^{mn}}{\mid C\mid}=q^{\sum \limits ^{\eta}_{j=1}d_jt_j}$.
Hence $C^{\perp}=D$.
\end{proof}
%=================================================
%============================
Now we can find the parity check matrix of the code $C$, as the following example shows.
\begin{ex}\label{ex2}
Suppose that $\mathcal{S}=\frac{\F_2[x,y]}{\langle x^3-1 , y^3-1\rangle}$ and $C$ is the code defined in Example
\ref{ex1}, that is,  $C= \ll (1+y+y^2)(x^2+x+1), (y+x)(x+1)\gg$. Hence
\begin{eqnarray*}
  h_1(x,y) &=& \frac{y^3-1}{y^2+y+1}=y+1\ \ \mathrm{in}\ \frac{\F_2[x]}{\langle x+1 \rangle}[y] \\
  h_1^{\sharp}(x,y) &=& \frac{y(\frac{1}{y}+1)}{0+1}=1+y \\
  h_2(x,y) &=& \frac{y^3-1}{y+x}=y^2+xy+x+1\ \ \mathrm{in}\ \frac{\F_2[x]}{\langle x^2+x+1 \rangle}[y] \\
 h_2^{\sharp}(x,y) &=& \frac{y^2 x(\frac{1}{y^2}+\frac{1}{x}\frac{1}{y}+\frac{1}{x}+1)}{x+1}=(y+x+1)(y+1) \\
\end{eqnarray*}
 We have $C^{\perp}= \ll (1+y)(x^2+x+1), (y+1)(y+x+1)(x+1)\gg$. Also, the parity check matrix of $C$ has the following form.
\begin{equation*}
%\mathbf{\mathbf
{H}= \left[
\begin{array}{c}
(1+y)(x^2+x+1)\\
y(1+y)(x^2+x+1)\\
(y^2+xy+x+1)(x+1)\\
x(y^2+xy+x+1)(x+1)\\
\end{array} \right]
=\left[
\begin{array}{ccccccccc}
1& 1& 1 & 1 & 1& 1& 0& 0 &0\\
0& 0 &0 & 1 & 1&1& 1& 1 & 1\\
1& 0& 1 & 0 & 1& 1& 1& 1 &0\\
1& 1 &0& 1 & 0 & 1 & 0 & 1& 1\\
\end{array} \right].
\end{equation*}
Every two columns of $H$ are linearly independent and there are $3$ linearly dependent columns in $H$. So
$d_{\min}(C)=3$.
\end{ex}
%=================================================Self Dual
%=================================================

Next we study when $C$ is self-dual, that is, $C=C^\perp$. To see why it is important to study and find self-dual codes see for example \cite[Section 3]{me}.
Note that if $C$ is self-dual, then it is both $(\lambda,
\delta)$-constacyclic and $(\lambda^{-1}, \delta^{-1})$-constacyclic.
\begin{lem}\label{diff delta}
Suppose that $C$ is both $(\lambda, \delta_1)$-constacyclic and $(\lambda', \delta_2)$-constacyclic, where
$\delta_1\neq \delta_2$ (but we may have $\lambda'=\lambda$). Then there exists a (one-dimensional)
$\lambda$-constacyclic code $C_0$ of length $n$ such that
 $$C=C_0^m= \{(\bc_0, \ldots, \bc_{m-1})|\forall 0\leq i<m: \quad \bc_i\in C_0\}.$$
Also in this case, $C$ is $(\lambda, \delta)$-constacyclic for every $\delta\in \F$. A similar result holds for
$(\lambda_1,\delta)$- and $(\lambda_2,\delta')$-constacyclic codes.
\end{lem}
\begin{proof}
Let $C_0$ be the set of $\bc\in \F^n$ such that $\bc$ appears as the $i$-th entry of a codeword of $C$. Note that since
$C$ is $\F$-linear and constacyclic, $C_0$ is independent of the choice of $i$. Also it is routine to see that $C_0$ is
a $\lambda$-constacyclic code of length $n$ and $C\subseteq C_0^m$. Conversely, if $\bc\in C_0$, say $u=(\bc_0,
\ldots,\bc_{m-1})\in C$ with $\bc=\bc_{m-1}$, then as $C$ is both  $(\lambda, \delta_1)$- and  $(\lambda',
\delta_2)$-constacyclic, $v=(\delta_1\bc, \ldots, \bc_{m-2})$ and $w=(\delta_2\bc, \ldots, \bc_{m-2})$ are in $C$ and
hence $v-w=((\delta_1-\delta_2) \bc, 0, \ldots, 0)\in C$. Thus by $F$-linearity and being constacyclic, we see that
$(0, \ldots, 0, \bc, 0, \ldots, 0)\in C$, where the position of $\bc$ is arbitrary. Again as $\bc\in C_0$ was chosen
arbitrarily and by linearity, we deduce that $C_0^m\subseteq C$.
\end{proof}

\begin{prop}\label{diff delta dual}
Suppose that $C$ is a $(\lambda, \delta)$-constacyclic code of length $nm$ and $\delta^2\neq 1$. Then $C$ is self-dual
\ifof $C=C_0^m$ for some self-dual $\lambda$-constacyclic code $C_0$ of length $n$. A similar result holds when instead
of $\delta^2\neq 1$ we assume $\lambda^2\neq 1$.
\end{prop}
\begin{proof}
By \cite[Propositin 2.2]{rajabi}, $C=C^\perp$  is a $(\lambda^{-1},\delta^{-1})$-consta\-cyc\-lic code and as
$\delta^{-1}\neq \delta$, we can apply \ref{diff delta} to see that $C=C_0^m$. It is routine to check that $C_0^m$ is
self-dual \ifof $C_0$ is self-dual.
\end{proof}
%=================================================

Thus to see when a constacyclic code $C$ is self-dual, it remains to consider the case that $\lambda^2=\delta^2=1$. So
assume that $\lambda^2=\delta^2=1$ and $C$ is a $(\lambda , \delta)$-consta\-cyc\-lic code of length $nm$. Also suppose
that $f_i^\sharp=f_i$ for all of the irreducible factors $f_i$ of $x^n-\lambda$. This is because in this case,
$\prod_{i\neq j}f_i^\sharp=\prod_{i\neq j}f_i$ for all $j$ and hence according to \ref{p4}, $C$ is self-dual \ifof
$g_j(x,y)=h_j^\sharp(x,y)$ for all $j$. Note that since $\delta=\pm 1$, we have $y^m-\delta=(y^{m'}-\delta)^{p^s}$. Let
in $K_j[y]$, $y^{m'}-\delta =\prod_{l=1}^{t_j}h_{jl}(x,y)$, where $h_{jl}(x,y), 1\leq l \leq t_j,$ are monic
irreducible coprime polynomials in $K_j[y]$. Assume that $h_{jl}(x,y)=h_{jl}^{\sharp}(x,y)$  for $1\leq l \leq a_j$ and
$h_{jl}^\sharp \neq h_{jl}$ for $a_j<l$. Since $(y^{m'}-\delta)^\sharp=y^{m'}-\delta^{-1}= y^{m'}-\delta$, so for each
$1\leq l\leq t_j$, we have $h_{jl}^\sharp= h_{jl'}$ for some $1\leq l'\leq t_j$. Thus we can suppose that

\begin{equation}\label{e21}
y^{m'}-\delta
=\prod_{l=1}^{a_j}h_{jl}(x,y)\prod_{l=a_j+1}^{b_j}h_{jl}(x,y)\prod_{l=a_j+1}^{b_j}h_{jl}^{\sharp}(x,y).
\end{equation}
 %============================================================================

 \begin{thm}\label{t11}
Let $p=2$, $s>0$, $f_i^{\sharp}(x)=f_i(x)$ for all $i$, and
   $$C=\ll g_1(x,y)\prod_{i\neq 1} f_i(x),g_2(x,y)\prod_{i\neq 2} f_i(x),\ldots,g_{\eta}(x,y)\prod_{i\neq \eta} f_i(x) \gg$$
be a $(\lambda , \delta)$-consta\-cyc\-lic code of length $n(2^sm')$ over $\F$, where $ \lambda^2=\delta^2=1$. The code
$C$ is self-dual if and only if for every $j$,
  \begin{equation}\label{e20}
     g_j(x,y)=\prod_{l=1}^{a_j}h_{jl}^{2^{s-1}}(x,y)\prod_{l=a_j+1}^{b_j}h_{jl}^{\alpha_{jl}}(x,y)
   \prod_{l=a_j+1}^{b_j}(h_{jl}^{\sharp })^{2^s-\alpha_{jl}}(x,y),
  \end{equation}
   for some $\alpha_{jl}$, $0\leq \alpha_{jl} \leq 2^s$.
 \end{thm}
 \begin{proof}
Assume that $C=C^{\perp}$. So for any $j$, we have $g_j(x,y)= h_j^{\sharp}(x,y)$ or equivalently
$g_j^\sharp(x,y)=h_j(x,y)$. This means,
   $$g_j(x,y)=\frac{y^m-\delta}{h_j(x,y)} =\frac{y^m-\delta}{g_j^{\sharp}(x,y)}.$$
   If
   $$ g_j(x,y)=\prod_{l=1}^{a_j}h_{jl}^{\gamma_{jl}}(x,y)\prod_{l=a_j+1}^{b_j}h_{jl}^{\alpha_{jl}}(x,y)
   \prod_{l=a_j+1}^{b_j}(h_{jl}^{\sharp })^{\beta_{jl}}(x,y),$$
   then
$$ g_j^\sharp(x,y)= \prod_{l=1}^{a_j}h_{jl}^{\gamma_{jl}}(x,y)\prod_{l=a_j+1}^{b_j}(h_{jl}^\sharp)^{\alpha_{jl}}(x,y)
   \prod_{l=a_j+1}^{b_j}h_{jl}^{\beta_{jl}}(x,y).$$ Hence
   \begin{eqnarray*}
    y^m-\delta &=&  g_j(x,y)g_j^{\sharp}(x,y) \\
     &=& \prod_{l=1}^{a_j}h_{jl}^{2\gamma_{jl}}(x,y)
   \prod_{l=a_j+1}^{b_j}h_{jl}^{\alpha_{jl}+\beta_{jl}}(x,y)
   \prod_{l=a_j+1}^{b_j}(h_{jl}^{\sharp })^{\alpha_{jl}+\beta_{jl}}(x,y).
   \end{eqnarray*}
 So $2\gamma_{jl}= \alpha_{jl}+\beta_{jl}=2^s$. Therefore, $\gamma_{jl}=2^{s-1}$ and $ \alpha_{jl}=2^s-\beta_{jl}$.
 Conversely, suppose that $g_j(x,y)$ is of the form \eqref{e20}. So
 $$ h_j(x,y)=\frac{y^m-\delta}{g_j(x,y)}=\prod_{l=1}^{a_j}h_{jl}^{2^{s-1}}(x,y)
 \prod_{l=a_j+1}^{b_j}h_{jl}^{2^s-\alpha_{jl}}(x,y)
   \prod_{l=a_j+1}^{b_j}(h_{jl}^{\sharp })^{\alpha_{jl}}(x,y).$$
   Hence
   $$h_j^{\sharp}(x,y)=\prod_{l=1}^{a_j}h_{jl}^{2^{s-1}}(x,y)
 \prod_{l=a_j+1}^{b_j}(h_{jl}^{\sharp })^{2^s-\alpha_{jl}}(x,y)
   \prod_{l=a_j+1}^{b_j}h_{jl}^{\alpha_{jl}}(x,y)=g_j(x,y).$$
   Thus $C=C^{\perp}$.
 \end{proof}
 %=================================================================
\begin{ex}\label{ex4}
Suppose that $\mathcal{S}=\frac{\F_2[x,y]}{\langle x^3-1 , y^{12}-1\rangle}$. We have
 \begin{eqnarray*}
x^3-1 &=&(x+1)(x^2+x+1)\ \ \ \ \ \mathrm{in}\ \F_2[x]\\
y^{12}-1&=&(y+1)^4(y^2+y+1)^4\ \ \ \ \ \mathrm{in}\ \frac{\F_2[x]}{\langle x+1\rangle}[y]\\
y^{12}-1&=&(y+1)^4(y+x)^4(y+x+1)^4\ \ \ \ \ \mathrm{in}\ \frac{\F_2[x]}{\langle x^2+x+1\rangle}[y].
\end{eqnarray*}
Let  $C=\ll (y^2+y+1)^2 (y+1)^2(x^2+x+1), (y+1)^2(y+x)^3(y+x+1)(x+1)\gg $
be a 2D consta\-cyc\-lic code of length 36 over $\F_2$. By Theorem \ref{t11}, $C$  is a self-dual code. But the code
 $$D= \ll (y^2+y+1) (y+1)^2(x^2+x+1), (y+1)^2(y+x)^2(y+x+1)(x+1)\gg $$ is not self-dual. We have $$D^{\perp}=\ll (y^2+y+1)^3 (y+1)^2(x^2+x+1), (y+1)^2(y+x)^3(y+x+1)^2(x+1)\gg. $$
   Note that $(y+x)^{\sharp}=y+x+1$.
\end{ex}
%=============================================================
The following theorem shows that, for an odd prime number $p$ the existence of self-dual one-sided repeated root  2D
consta\-cyc\-lic codes  depends on the factorization of $y^{m'}-\delta$.
\begin{thm}\label{t12}
Assume that $f_i^{\sharp}(x)=f_i(x)$, for all $i$ and $ \lambda^2=\delta^2=1$. Let $p$ be an odd prime number or $s=0$.
There exists a self-dual 2D $(\lambda, \delta)$-consta\-cyc\-lic code of length $nm=n(p^sm')$ over $\F$ if and only if
in \eqref{e21}, $a_j=0$ for all $j$. In this case, a code
  $$C=\ll g_1(x,y)\prod_{i\neq 1} f_i(x),g_2(x,y)\prod_{i\neq 2} f_i(x),\ldots,g_{\eta}(x,y)\prod_{i\neq \eta} f_i(x) \gg$$
  is self-dual if and only if
   $$  g_j(x,y)=\prod_{l=1}^{b_j}h_{jl}^{\alpha_{jl}}(x,y)
   \prod_{l=1}^{b_j}(h_{jl}^{\sharp })^{p^s-\alpha_{jl}}(x,y),$$
   for some $\alpha_{jl}$, $0\leq \alpha_{jl} \leq p^s$.
 % $$y^{m'}-\delta =\prod_{l=1}^{b_j}h_{jl}(x,y)\prod_{l=1}^{b_j}h_{jl}^{\sharp}(x,y),$$
  %where $h_{jl}(x,y)\neq h_{jl}^{\sharp}(x,y)$ for any $l$, $1\leq l \leq b_j.$
\end{thm}
\begin{proof}
Assume that $p$ is an odd prime number and
   $$y^{m'}-\delta =\prod_{l=1}^{b_j}h_{jl}(x,y)\prod_{l=1}^{b_j}h_{jl}^{\sharp}(x,y).$$
   Consider the polynomials $g_j(x,y)=\prod_{l=1}^{b_j}h_{jl}^{\alpha_{jl}}(x,y)
   \prod_{l=1}^{b_j}(h_{jl}^{\sharp })^{p^s-\alpha_{jl}}(x,y)$, for $j$, $1\leq j \leq \eta$. Thus
    $$ h_j(x,y)=\frac{y^m-\delta}{g_j(x,y)}=\prod_{l=1}^{b_j}h_{jl}^{p^s-\alpha_{jl}}(x,y)
   \prod_{l=11}^{b_j}(h_{jl}^{\sharp })^{\alpha_{jl}}(x,y).$$
   Hence
   $$h_j^{\sharp}(x,y)=\prod_{l=1}^{b_j}(h_{jl}^{\sharp })^{p^s-\alpha_{jl}}(x,y)
   \prod_{l=1}^{b_j}h_{jl}^{\alpha_{jl}}(x,y)=g_j(x,y).$$
So $C=\ll g_1(x,y)\prod_{i\neq 1} f_i(x),\ldots,g_{\eta}(x,y)\prod_{i\neq \eta} f_i(x) \gg$ is a self-dual code.
Conversely, suppose that $y^{m'}-\delta $ is of the form \eqref{e21} with $a_j\neq 0$ for some $j$ and $C=\ll
g_1(x,y)\prod_{i\neq 1} f_i(x),\ldots,g_{\eta}(x,y)\prod_{i\neq \eta} f_i(x) \gg$ is a self-dual code. For any $j$, we
have
  \begin{equation*}
     g_j(x,y)=\prod_{l=1}^{a_j}h_{jl}^{\gamma_{jl}}(x,y)\prod_{l=a_j+1}^{b_j}h_{jl}^{\alpha_{jl}}(x,y)
   \prod_{l=a_j+1}^{b_j}(h_{jl}^{\sharp })^{\beta_{jl}}(x,y),
  \end{equation*}
for some $\alpha_{jl} $, $\gamma_{jl} $ and $\beta_{jl} $, $0\leq \alpha_{jl}, \gamma_{jl}, \beta_{jl} \leq p^s$. Now,
similar to the proof Theorem \ref{t11}, we have $\alpha_{jl}+ \beta_{jl} =2\gamma_{jl}=p^s$. Since $p$ is an odd prime
number or $p^s=1$, this is impossible. So $a_j=0$ and $g_j(x,y)$ has the claimed form for all $j$.
\end{proof}
%===============================================
%\begin{cor}
 %  Let $p$ be an odd prime number, $ \lambda^2=\delta^2=1$ and there exists an self-dual 2D consta\-cyc\-lic code $C$ of length $mn$ over $\F$. Then
  % $$C=\ll g_1(x,y)\prod_{i\neq 1} f_i(x),g_2(x,y)\prod_{i\neq 2} f_i(x),\ldots,g_{\eta}(x,y)\prod_{i\neq \eta} f_i(x) \gg,$$
  % where the polynomials $g_j(x,y)$ are of the form
  % $$  g_j(x,y)=\prod_{l=1}^{b_j}h_{jl}^{\alpha_{jl}}(x,y)
  % \prod_{l=1}^{b_j}(h_{jl}^{\sharp })^{p^s-\alpha_{jl}}(x,y)$$.
%\end{cor}
%======================================
\begin{ex}\label{ex5}
 Suppose that $\mathcal{S}=\frac{\F_9[x,y]}{\langle x^2-1 , y^2+1\rangle}$. We have
 $$ x^2-1=(x-1)(x+1) \ \mathrm{and}\ y^2+1=(y+\alpha)(y+\alpha^{-1}),$$
 where $\alpha^2=-1$ in $\F_9$. Since
 $$ (x+1)^{\sharp}=x+1 \ \mathrm{and}\ (x-1)^{\sharp}=x-1,$$
 we can use Theorem \ref{t12}, to determine self-dual codes over $\F_9$. By Theorem \ref{t12},
 $$D_1=\ll (y+\alpha)(x-1), (y+\alpha^{-1})(x+1)\gg$$
 and
 $$D_2=\ll (y+\alpha^{-1})(x-1), (y+\alpha)(x+1)\gg$$
are self-dual codes. Note that we also have $\mathcal{S} \cong \frac{\F_9[x,y]}{\langle x^2+1 , y^2-1\rangle}$, but if
we view $S$ as this ring, then since $(x+\alpha)^{\sharp}\neq (x+\alpha)$, we can not use Theorem \ref{t12}.
\end{ex}
%===================================
%If $p=2$, then $\delta^2=1 \iff \delta=1$ and $y-1|y^m-1$. Since $y-1=(y-1)^\sharp$, we the following corollary follows
%from \ref{t12} immediately.
%\begin{cor}
%Suppose that $p=2$ and for every monic irreducible factor $f$ of $x^n+1$, we have $f=f^\sharp$. Then no simple root
%self-dual binary  2D constacyclic code of length $nm$ exists.
%\end{cor}

%@@@@@@@@@@@@@@@@@@@@@@@@@@@@@@@@@@@@@@@@@@@@@@@@@@@@@@@@@@@@@@@@
\section{Asymptotic badness of one-sided repeated-root  2D consta\-cyc\-lic codes}
%@@@@@@@@@@@@@@@@@@@@@@@@@@@@@@@@@@@@@@@@@@@@@@@@@@@@@@@@@@@@@@@@

Theorem \ref{p2} enables us to state and prove 2D versions of some known results on repeated root (one dimensional)
cyclic and constacyclic codes. One such result is \cite[Theorem 4]{casta}, which states that repeated-root cyclic codes
cannot be asymptotically better than simple root cyclic codes. In this section, similar to \cite{casta},
we relate the minimum distance of a one-sided repeated-root 2D consta\-cyc\-lic code to some
simple-root 2D consta\-cyc\-lic codes and from this, we deduce that one-sided repeated-root $(\lambda ,
\delta)$-consta\-cyc\-lic codes can not be asymptotically better than simple-root 2D consta\-cyc\-lic codes. It should
be mentioned that the ideas and techniques used here are quite similar to those used in sections $III$ and $IV$ of
\cite{casta} and hence we do not state most of the proofs here.

Note that for any $\delta \in \F$, there exists a $\delta_0 \in \F$ with the property  that $\delta=\delta_0^{p^s}$
(see, for example \cite[Lemma 3.1]{batoul}). Let $C=\ll g_1(x,y)\prod  _{i\neq 1} f_i(x),\ldots ,g_{\eta}(x,y)\prod
_{i\neq \eta} f_i(x) \gg$ be a $(\lambda , \delta)$-consta\-cyc\-lic code over $\F$. Recall that $C$ is an ideal of
$\mathcal{S}=\frac{\F[x,y]}{\langle x^{n}-\lambda ,y^{m}-\delta \rangle }$, where $\gcd(n,p)=1$, $m=m'p^s$ with
$\gcd(m',p)=1$ and $x^n-\lambda =\prod_{i=1}^{\eta}f_i(x)$ where $f_i(x), 1\leq i \leq \eta,$ are monic irreducible
coprime polynomials in $\F[x]$ and $ \deg f_i=d_i$. Suppose that in $K_j[y]$ we have $y^{m'}-\delta_0
=\prod_{l=1}^{t_j}h_{jl}(x,y)$, where $h_{jl}(x,y), 1\leq l \leq t_j,$ are monic irreducible coprime polynomials in
$K_j[y]$. Then as $g_j(x,y)$ is a divisor of $(y^{m'}-\delta_0)^{p^s}=y^m-\delta$ in $K_j[y]$, we have $g_j(x,y)=\prod
_{l=1}^{t_j}h_{jl}^{\alpha _{jl}}(x,y)$ for some $\alpha _{jl}$, $0\leq \alpha _{jl} \leq p^s$.
%==============================================
For $0\leq t \leq p^s-1$, let $\bar{g}_{jt}(x,y)$ be the product of those irreducible factors $h_{jl}(x,y)$ such that
$\alpha _{jl}>t$ in $g_j(x,y)$. So
\begin{center}
  $\bar{C}_t=\ll \bar{g}_{1t}(x,y)\prod \limits _{i\neq 1} f_i(x),\bar{g}_{2t}(x,y)\prod \limits _{i\neq 2} f_i(x),\ldots ,\bar{g}_{\eta t}(x,y)\prod \limits _{i\neq \eta} f_i(x) \gg$
\end{center}
is an ideal of $\bar{\mathcal{S}}=\frac{\F[x,y]}{\langle x^{n}-\lambda ,y^{m'}-\delta_0 \rangle }$ and a simple root 2D
constacyclic code. Note that, if $t'\geq t$, then $\bar{C}_{t'}\supseteq \bar{C}_{t}$.
%==================================================
%==================================================================
\begin{ex}\label{ex3}
Suppose that $\mathcal{S}=\frac{\F_2[x,y]}{\langle x^3-1 , y^{12}-1\rangle}$ (see Example \ref{ex4}).
Let $C=\ll (y^2+y+1)(x^2+x+1), (y+1)(y+x)^3(x+1)\gg $ be a 2D consta\-cyc\-lic code of length 36 over $\F_2$. So
\begin{eqnarray*}
 \nonumber % Remove numbering (before each equation)
  \bar{C}_0 &=&\ll  (y^2+y+1)(x^2+x+1), (y+1)(y+x)(x+1)\gg\\
  \bar{C}_1 &=&  \ll (x^2+x+1), (y+x)(x+1)\gg\\
  \bar{C}_2 &=& \ll  (x^2+x+1), (y+x)(x+1)\gg\\
   \bar{C}_3 &=& \ll  (x^2+x+1),(x+1)\gg = \lg 1 \rg
\end{eqnarray*}
are ideals of $\frac{\F_2[x,y]}{\langle x^3-1 , y^{3}-1\rangle}$, that is, simple-root 2D  consta\-cyc\-lic codes of
length 9 over $\F_2$. \qed
\end{ex}
%================================================================
We start by the following lemma.
\begin{lem}\label{l2}
Let $C=\ll g_1(x,y)\prod \limits _{i\neq 1} f_i(x),\ldots ,g_{\eta}(x,y)\prod \limits _{i\neq \eta} f_i(x) \gg$ be a
$(\lambda , \delta)$-consta\-cyc\-lic code over $\F$ and $\bar{v}_{t}(x,y)$ be a non-zero polynomial in $\bar{C}_{t}$.
Then the polynomial
\begin{equation}
 \bar{c}_{t}(x,y)=(y^{m'}-\delta _0)^t \bar{v}_{t}^{p^s}(x,y) \mod (y^m-\delta , x^n-\lambda)\label{e1}
\end{equation}
is a non-zero polynomial in $C$.
\end{lem}
\begin{proof}
Since $\bar{v}_{t}(x,y)$ is a non-zero polynomial in $\bar{C}_{t}$, so \\
 \centerline{$ \bar{v}_{t}(x,y)=\sum \limits ^{\eta}_{j=1}\l( A_j(x,y)\bar{g}_{jt}(x,y)\prod \limits _{i\neq j} f_i(x)\r)$,  where $A_j(x,y)\in \bar{\mathcal{S}}$.}\\
 When $j\neq l$, we have
 \begin{center}
 $(\prod \limits _{i\neq l} f_i(x))(\prod \limits _{i\neq j} f_i(x))=0\  \mod (y^m-\delta , x^n-\lambda)$.
 \end{center}
 Hence
 \begin{center}
   $ \bar{v}_{t}^{p^s}(x,y)=\sum \limits ^{\eta}_{j=1}B_j(x,y)\bar{g}_{jt}^{p^s}(x,y)(\prod \limits _{i\neq j} f_i(x))^{p^s}$\ where $ B_j(x,y)\in \mathcal{S}$.
 \end{center}
 Thus\\
 \centerline{$ \bar{v}_{t}^{p^s}(x,y)\in \ll g_{1t}^{p^s}(x,y)(\prod \limits _{i\neq 1} f_i(x))^{p^s},\ldots ,g_{\eta t}^{p^s}(x,y)(\prod \limits _{i\neq \eta} f_i(x))^{p^s} \gg$.}\\
 Now, completely similar to the first part of \cite[Lemma 1]{casta}, this lemma can be proved.
\end{proof}
%======================================================
In what follows, by $w(c)$ and $d_{\min}(C)$ we mean the Hamming weight of a codeword $c$ and the minimum Hamming
distance of a code $C$.  Let $P_t=w((y^{m'}-\delta _0)^t)$. By Lemma \ref{l2} and using the techniques used in the second part of the proof of
\cite[Lemma 1]{casta}, we have the following lemma.
\begin{lem}\label{l3}
Let $C$ be a $(\lambda , \delta)$-consta\-cyc\-lic code over $\F$. Then $d_{\min}(C) \leq P_t d_{\min}(\bar{C}_{t})$,
for any $0\leq t\leq p^s-1$.
\end{lem}
  %=======================================================================
\par Consider the set
\begin{equation}\label{e3}
  T=\{t<p^s \mid P_t<P_{t'} \ \ \ \mathrm{for\ any}\ t'\in \{t+1,\ldots ,p^s-1\}\}.
\end{equation}
Note that $p^s-1\in T$ and for any $0\leq t\leq p^s-1$ set
\begin{equation}\label{e4}
  \bar{t}=\min \{ t'\in T \mid t'\geq t \}.
\end{equation}

%=============================================
\begin{lem}\label{l4}
Let $C$ be a $(\lambda , \delta)$-consta\-cyc\-lic code of length $nm$ over $\F$. Assume that $c(x,y)\in C$ and
$c(x,y)=(y^{m'}-\delta _0)^t v(x,y)$, where $y^{m'}-\delta _0\nmid v(x,y)$ in $\F[x,y]$. Then the polynomial
 \begin{center}
 $\bar{c}_{\bar{t}}(x,y)=(y^{m'}-\delta _0)^{\bar{t}} \bar{v}^{p^s}(x,y) \mod{(y^m-\delta, x^n-\lambda)}$,
 \end{center}
where $\bar{v}(x,y)= v(x,y) \mod( y^m-\delta, x^n-\lambda)$, is a non-zero polynomial of $C$ and satisfies
$w(\bar{c}_{\bar{t}}(x,y))\leq w(c(x,y))$.
\end{lem}
\begin{proof}
The proof is similar to the proof of \cite[Lemma 2]{casta}, but with some minor modifications. To show how these
modifications can be accounted (in this and also other omitted proofs), we state the complete proof here. Since $c(x,y)\in
C$, $c(x,y)=\sum ^{\eta}_{j=1}A_j(x,y)g_{j}(x,y)\prod
 _{i\neq j} f_i(x)$,  for some $A_j(x,y)\in \mathcal{S}$. So in $\mathcal{S}$,
  \centerline{$\sum \limits ^{\eta}_{j=1}A_j(x,y)g_{j}(x,y)\prod \limits _{i\neq j} f_i(x)=(y^{m'}-\delta _0)^t v(x,y)$.}
Thus in $\F[x,y]$,
  \begin{center}
    $\sum \limits ^{\eta}_{j=1}A_j(x,y)g_{j}(x,y)\prod \limits _{i\neq j} f_i(x)=(y^{m'}-\delta _0)^t v(x,y)+(y^m-\delta )D+(x^n-\lambda)E$,
  \end{center}
for some $D,E\in \F[x,y]$. Now, in $K_j[y]$, we have
   \begin{center}
    $A_j(x,y)g_{j}(x,y)\prod \limits _{i\neq j} f_i(x)=(y^{m'}-\delta _0)^t v(x,y)+(y^m-\delta )D$.
  \end{center}
If for some $0\leq l \leq t_j$, $\alpha _{jl}>t$, then $h_{jl}^{t+(\alpha _{jl} -t)}(x,y)\mid g_j(x,y)$ and
$h_{jl}^{t+(\alpha _{jl} -t)}(x,y)\mid y^m-\delta $. Thus in $K_j[y]$,
  \begin{center}
   $h_{jl}^{t+(\alpha _{jl} -t)}(x,y)\mid ( y^{m'}-\delta_0)^t v(x,y) $.
  \end{center}
So $h_{jl}(x,y)\mid v(x,y)$. Hence $\bar{g}_{jt}(x,y)\mid v(x,y)$ in $K_j[y]$ for each $j$. Thus $\bar{v}(x,y)\in
\bar{C}_t$. Since $y^{m'}-\delta _0\nmid v(x,y)$ and $\deg_x v(x,y)<n$, $\bar{v}(x,y)\neq 0$ in $\bar{\mathcal{S}}$.
Since $\bar{t}\geq t$, $\bar{C}_{t}\subseteq \bar{C}_{\bar{t}}$. So $\bar{v}(x,y)\in \bar{C}_{\bar{t}}$. Hence by Lemma
\ref{l2},
  \begin{equation}\label{e9}
    0\neq \bar{c}_{\bar{t}}(x,y)=(y^{m'}-\delta _0)^{\bar{t }}\bar{v}^{p^s}(x,y)\in C.
  \end{equation}
 There exist polynomials $v_i\in \F[x,y]$ and $v_{ij}\in \F[y]$ such that
  \begin{equation*}%\label{e5}
   v(x,y)=\sum \limits ^{m'-1}_{i=0}y^i  v_i(x,y^{m'})
 \quad\text{and}  \quad
   v_i(x,y^{m'})=\sum \limits ^{n-1}_{j=0}x^j v_{ij}(y^{m'}).
  \end{equation*}
  %and $v_{ij}(y^{m'})=\sum \limits ^{p^s-1}_{e=0} a_{(i+em')j} (y^{m'})^e $ for some $e\in \F$.
  Hence
  \begin{equation*}
   v(x,y)=\sum \limits ^{m'-1}_{i=0}\sum \limits ^{n-1}_{j=0} x^j y^i v_{ij}(y^{m'}),
  \end{equation*}
  and
  \begin{equation*}
   \bar{v}(x,y)=\sum \limits ^{m'-1}_{i=0}\sum \limits ^{n-1}_{j=0} x^j y^i v_{ij}(\delta _0) \mod (y^{m'}-\delta _0, x^n-\lambda).
  \end{equation*}
Let $N_{\bar{\nu}}=w(\bar{v}(x,y))$, which according to the above formula equals the number of $v_{ij}$'s with
$v_{ij}(\delta_0)\neq 0$.  Now we have
  \begin{align*}
    w(c) =& w((y^{m'}-\delta _0)^t v(x,y)) =
    w((y^{m'}-\delta _0)^t \sum \limits ^{m'-1}_{i=0}y^i  v_i(x,y^{m'}))\\
     =& w( \sum \limits ^{m'-1}_{i=0}y^i (y^{m'}-\delta _0)^t v_i(x,y^{m'}))=
      \sum \limits ^{m'-1}_{i=0}w( (y^{m'}-\delta _0)^t v_i(x,y^{m'}))\\
      =&  \sum \limits ^{m'-1}_{i=0} w( (y^{m'}-\delta _0)^t \sum \limits ^{n-1}_{j=0}x^j v_{ij}(y^{m'}) )\\
      =& \sum \limits ^{m'-1}_{i=0} w(  \sum \limits ^{n-1}_{j=0}x^j(y^{m'}-\delta _0)^t v_{ij}(y^{m'}) )\\
       =&  \sum \limits ^{m'-1}_{i=0} \sum \limits ^{n-1}_{j=0} w( (y^{m'}-\delta _0)^t v_{ij}(y^{m'}) )=
        \sum \limits ^{m'-1}_{i=0} \sum \limits ^{n-1}_{j=0} w( (y-\delta _0)^t v_{ij}(y) ).\\
  \end{align*}
If $v_{ij}(y)\neq 0$, since $(y-\delta_0)^t v_{ij}(y)$ can be written as $\sum_{t'=t}^{p^s-1} b_{t'} (y-\delta_0)^{t'}$
with $b_i\in \F$ and by \cite[Theorem 6.1]{massy}, we have
\begin{eqnarray*}
w((y-\delta_0)^tv_{ij}(y) )\geq  \sum \limits ^{m'-1}_{i=0} \sum \limits ^{n-1}_{j=0} \min _{p^s>t' \geq t}w( (y-\delta _0)^{{t'}})=P_{\bar t}.
\end{eqnarray*}
Thus $w(c(x,y)) \geq P_{\bar{t}}N_{\nu}$, where $N_{\nu}$ is the number of nonzero $v_{ij}(y)$.  Also $
w(\bar{c}_{\bar{t}}(x,y))\leq P_{\bar{t}}w(\bar{v}^{p^s}(x,y))= P_{\bar{t}}w(\bar{v}(x,y))$ by \eqref{e9} and so $
w(\bar{c}_{\bar{t}}(x,y))\leq {\bar{t}}N_{\bar{\nu}}.$ If $v_{ij}(\delta _0)\neq 0$, then $v_{ij}(y)\neq 0$. Hence
$N_{\nu}\geq N_{\bar{\nu}}$. Therefore,
    \begin{center}
      $w(\bar{c}_{\bar{t}}(x,y))\leq P_{\bar{t}}N_{\bar{\nu}}\leq P_{\bar{t}}N_{\nu}\leq w(c(x,y))$.
    \end{center}
    Hence  $w(\bar{c}_{\bar{t}}(x,y))\leq w(c(x,y))$.
\end{proof}
%===================================================
%Recall that, if $\mathbf{u}$ and $\mathbf{v}$ are strings of the same length over the same alphabet, the Hamming
%distance $d(\mathbf{u}, \mathbf{v})$ between $\mathbf{u}$ and $\mathbf{v}$ is the number of positions in which
%$\mathbf{u}$ and $\mathbf{v}$ differ. The minimum distance of a code $C$ is defined to be
%\begin{center}
%$d_{\min}(C)=\min \{d(\mathbf{u}, \mathbf{v})\mid \mathbf{u}, \mathbf{v} \in C, \mathbf{u}\neq \mathbf{v} \}$.
%\end{center}
%======================================================
Using Lemmas \ref{l2}--\ref{l4} instead of Lemmas 1 and 2 of \cite{casta}, one can prove the following result similar
to \cite[Theorem 1]{casta}.
\begin{prop}\label{p5}
Let $C$ be a $(\lambda , \delta)$-consta\-cyc\-lic code of length $nm$ over $\F$. Then $d_{\min}(C)=P_{
\bar{t}}d_{\min}(\bar{C}_{\bar{t}})$, for some $\bar{t}\in T$.
\end{prop}
%\begin{proof}
%The proof process is similar to \cite[Theorem 1]{casta}.
%\end{proof}
%===================================================
%@@@@@@@@@@@@@@@@@@@@@@@@@@@@@@@@@@@@@@@@@@@@@@@@@@@@@@@@@@@@@@@@@@@@@@@@@@@@@@2
%============================================================
%=====================================
Also the next theorem can be proved mutatis mutandis to \cite[Theorem 2]{casta}.
\begin{thm}\label{p6}
Let $C$ be a one-sided repeated-root $(\lambda , \delta)$-consta\-cyc\-lic code of length $nm$ over $\F$. Then there
exists a simple-root $(\lambda , \delta)$-consta\-cyc\-lic code $\hat{C}$ over $\F$ of length $nm'$ with both rate and
relative minimum distance at least as large as the corresponding values for $C$.
\end{thm}
%\begin{proof}
%  In Lemma \ref{l3}, consider $t=p^s-1$ and $\bar{C}_{p^s-1}=\hat{C}$. So $d_{\min}(C)\leq P_{ p^s-1}d_{\min}(\hat{C})$. Moreover, $P_{ p^s-1}=w((y^{m'}-\delta_0)^{ p^s-1})=p^s$. So
 % \begin{center}
 %  $ \frac{d_{\min}(C)}{mn}=\frac{d_{\min}(C)}{m'p^sn}\leq \frac{p^s d_{\min}(\hat{C})}{m'p^sn}=\frac{ d_{\min}(\hat{C})}{m'n}$.
 % \end{center}
%\end{proof}
Note that, since the code $\hat{C}$ is of length $nm'$ and $C$ is of length $nm'p^s$, Theorem \ref{p6} does not mean
that simple root 2D constacyclic codes are better than one-sided repeated-root 2D consta\-cyc\-lic codes.

%============================================================
Let $C=\ll g_1(x,y)\prod  _{i\neq 1} f_i(x),\ldots ,g_{\eta}(x,y)\prod  _{i\neq \eta} f_i(x) \gg$ be a $(\lambda ,
\delta)$-consta\-cyc\-lic code of length $nm$ over $\F$. Then $\dim(C)=mn-\sum ^{\eta}_{j=0}d_j t_j$, where $d_j=\deg
f_j$ and $t_j=\deg _y g_j$. In what follows, we assume that $C$ and $\ll g_j(x,y)\gg$ (as an ideal of
$\frac{K_j[y]}{\langle y^m-\delta \rangle})$, are of the rate $r$ and $r_j$ respectively. So we have the following
lemma.

%====================================
\begin{lem}\label{l5}
  For any $R\ (0<R<1)$, if $R<r$, then there exists an $l\ (1\leq l \leq \eta )$ such that $R<r_l$.
\end{lem}
\begin{proof}
 Note that,
  \begin{equation*}\label{e7}
r=\frac{mn-\sum \limits ^{\eta}_{j=0}d_j t_j}{mn}=\frac{\sum \limits ^{\eta}_{j=0}d_j (m-t_j)}{mn}=\frac{1}{n}\sum
\limits ^{\eta}_{j=0}\frac{m-t_j}{m}d_j=\frac{1}{n}\sum \limits ^{\eta}_{j=0}r_jd_j.
\end{equation*}
Suppose that for any $j$, $R\geq r_j$. So $\sum \limits ^{\eta}_{j=0}d_jR \geq \sum \limits ^{\eta}_{j=0}r_jd_j$. Hence
$nR\geq \sum \limits ^{\eta}_{j=0}r_jd_j.$ Thus $R\geq r $ which is impossible.
\end{proof}
%===========================================
Now, one can prove the counterparts of Lemma 3 and Theorems 3 and 4 of \cite{casta}, with completely similar arguments.
Here we just mention the statements.
%===========================================
\begin{lem}\label{l6}
Let $C$ be a $(\lambda , \delta)$-consta\-cyc\-lic code of length $nm$ over $\F$ and rate $r$. For any $R$ ($0<R<1$),
there exists a constant $\gamma(R)$ such that if $r>R$, then $d_{\min}(C)\leq p^{ \gamma(R)}nm'.$
\end{lem}
%\begin{proof}
% Use the Lemma \ref{l5} and the process of the proof of \cite[Lemma 3]{casta}.
%\end{proof}
%===========================================
\begin{prop}\label{p7}
Any sequence of $(\lambda , \delta)$-consta\-cyc\-lic codes $C_i$ over $\F$ of length $n_im_i$, where $m_i=m'_ip^{s_i}$
and $\gcd(n_i,p)=\gcd(m'_i,p)=1$, with rates $r_i>R>0$ such that $\liminf_{i\rightarrow \infty} s_i=\infty$, satisfies
  \begin{center}
    $\lim\limits_{i\rightarrow \infty}\frac{d_{\min}(C_i)}{m_in_i} =0$.
  \end{center}
\end{prop}
%===========================================
\begin{thm}\label{p8}
If there exist a sequence of $(\lambda , \delta)$-consta\-cyc\-lic codes $C_i$ over $\F$ of length $n_i m_i$, where
$m_i=m'_i p^{s_i}$ and $\gcd(n_i,p)=\gcd(m'_i,p)=1$, with rates $r_i>R>0$ such that
   \begin{center}
    $\liminf\limits_{i\rightarrow \infty} m_i n_i=\infty$,
     \end{center}
     and
     \begin{center}
     $\liminf\limits_{i\rightarrow \infty}\frac{d_{\min}(C_i)}{m_i n_i} =\Delta >0$,
       \end{center}
      then there exists a sequence of $(\lambda , \delta)$-consta\-cyc\-lic codes $\hat{C}_i$ over $\F$ of length $\hat{n}_i\hat{m}_i$ %, where $\hat{m}_i=\hat{m}'_i p^{\hat{s}_i}$
      and $\gcd(\hat{n}_i,p)=\gcd(\hat{m}_i,p)=1$, and of the rate $\hat{r}_i>R>0$ such that
  \begin{center}
    $\lim\limits_{i\rightarrow \infty} \hat{n}_i\hat{m}_i=\infty$
  \end{center}
  and
  \begin{center}
    $\liminf\limits_{i\rightarrow \infty}\frac{d_{\min}(\hat{C}_i)}{\hat{m}_i\hat{n}_i}\geq \Delta >0$.
  \end{center}
\end{thm}
%==================================
This theorem shows that if there exists an asymptotically good family of one-sided repeated-root 2D consta\-cyc\-lic
codes, then there exists an asymptotically good family of simple-root 2D consta\-cyc\-lic codes with similar or better
parameters.

%@@@@@@@@@@@@@@@@@@@@@@@@@@@@@@@@@@@@@@@@@@@@@@@@@@@@@@@@@@@@@@@@
\section{Conclusion}
%@@@@@@@@@@@@@@@@@@@@@@@@@@@@@@@@@@@@@@@@@@@@@@@@@@@@@@@@@@@@@@@@
In this paper, we studied the algebraic structure of one-sided repeated 2D cyclic and constacyclic codes, found their
generator and parity check matrices and also their duals. Moreover, we showed that if there exists an asymptotically good family of one-sided repeated-root
2D cyclic or consta\-cyc\-lic codes, then there exists an asymptotically good family of simple-root 2D consta\-cyc\-lic codes with similar or better
parameters.  We also mentioned some needed corrections to some known previous results on such codes.

Although our work suggests that one should not
look for asymptotically good families of codes in this class of codes, but we may have some optimal codes among one-sided repeated-root cyclic or constacyclic codes. We leave
their study to a future work. Another  problem which remains to be studied in future works is to find self-dual codes when the assumption made here (that is, $f_i^\sharp=f_i$)
does not hold.
\paragraph{Acknowledgement} The authors would like to thank Prof. H. Sharif of Shiraz University for his nice comments
and discussions.
%%%%%%%%%%%%%%%%%%%%%%%%%%%%%%%%%%%%%%%%%%%%%%%%%%%%%%%%%%%
                                

\begin{thebibliography}{99}
%%%%%%%%%%%%%%%%%%%%%%%%%%%%%%%%%%%%%%%%%%%%%%%%%%%%%%%%%%%


\bibitem{IEEE} C. Aguilar-Melchor, O. Blazy, J.-C. Deneuville, P. Gaborit and G. Z\'emor, \textit{Efficient encryption
    from random quasi-cyclic codes}, IEEE Trans. Inform. Theory, \textbf{64}(5) (2018), 3927--3943.

\bibitem{aty} M. F. Atiyah and I. G. Macdonald, \textit{Introduction to Commutative Algebra}, Addison-Wesley, London
    (1969).

\bibitem{Aydin} N. Aydin, N. Connolly and J. Murphee, \textit{New binary linear codes from quasi-cyclic codes and an
    augmentation   algorithm}, Appl. Algebra Eng. Commun. Comput., \textbf{28} (2017), 339--350.

\bibitem{batoul} A. Batoul and K. Guenda,  \textit{Some constacyclic codes over finite chain rings}, Advances in Math.
    of comm., \textbf{10} (2016), 683--694.


\bibitem{beygi} M. Beygi, S. Namazi and H. Sharif, \textit{Algebraic structures of constacyclic codes over finite chain
    rings
    and power series rings}, Iranian J. Sci. Tech. Trans. Sci., \textbf{43}(5) (2019), 2461--2476.

\bibitem{casta} G. Castagnoli, J. L. Massy, P. A. Schoeller and N. Seemann, \textit{On repeated-root cyclic codes},
    IEEE Trans. Inform. Theory, \textbf{37}(2) (1991), 337--342.

\bibitem{chen} E. Z. Chen, \textit{New binary h-generator quasi-cyclic codes by augmentation and new minimum distance
    bounds}, Des. Codes Cryptogr., \textbf{80} (2016), 1--10

%\bibitem{consta}  B. Chen, Y. Fan, L. Lin and H. Liu, \textit{Constacyclic codes over finite fields}, Finite Field
%    Appl.,  \textbf{18} (2012), 1217--1231.

\bibitem{Dinh} H. Q. Dinh, X. Wang and P. Maneejuk, \textit{On the Hamming distance of repeated-root cyclic codes of length $6ps$}, IEEE Access,
\textbf{8} (2020), 39946--39958.

\bibitem{Gun quasi} C. G\"uneri and F. \"Ozbudak, \textit{A relation between quasi-cyclic codes and 2-D cyclic codes},
    Fintie Fields Appl.,
    \textbf{18} (2012), 123--132.

\bibitem{Gun} C. G\"uneri and F. \"Ozbudak, \textit{Multidimensional cyclic codes and Artin-Schreier type hypersurfaces
    over    finite  fields}, Finite Fields Appl., \textbf{14} (2008), 44--58.

\bibitem{Gun3} C. G\"uneri,  B. \"Ozkaya and S. Say\i c\i \textit{On Linear Complementary Pair of $n$D Cyclic Codes}, IEEE Commun. Letters,
\textbf{22}(12) (2018), 2404--2406.

\bibitem{ik} T. Ikai, H. Kosako and Y. Kojima, \textit{Two-dimensional cyclic codes}, Trans. Inst. Electron. Commun.
    Eng. Jpn.,  \textbf{57} (1974), 279-286.


\bibitem{im} H. Imai, \textit{A theory of two-dimensional cyclic codes}, Inf. control., \textbf{34}(1) (1977), 1--21.

\bibitem{Kumar} R. Kumar and M. Bhaintwal, \textit{Repeated root cyclic codes over $\z_{p^2}+u\z_{p^2}$ and their Lee distances}, Cryptogr. Commun.,
\textbf{14} (2022), 551--557.

\bibitem{lidl} R. Lidl and H. Niederreiter, \textit{Finite Fields}, Cambridge University Press, Cambridge (1997).

\bibitem{ma} J. Ma and J. Luo, \textit{MDS symbol-pair codes from repeated-root cyclic codes}, Des. Codes Cryptogr., \textbf{90} (2022), 121--137.

\bibitem{massy} J. L. Massy, D. J. Costello and J. Justesen, \textit{Polynomial weights and code constructions}, IEEE
    Trans. Inform. Theory, \textbf{19}(1) (1973), 101--110.

\bibitem{nature}  M. Milicevic, C. Feng,  L.M. Zhang and  P. G. Gulak, \textit{Quasi-cyclic multi-edge LDPC codes for long-distance quantum cryptography},
 npj Quantum Inf, \textbf{4} (2018), article no. 21.

\bibitem{me} A. Nikseresht, \textit{Dual of codes over finite quotients of polynomial rings}, Finite Fields Appl., \textbf{45} (2017), 323--340.

\bibitem{rajabi} Z. Rajabi and K. Khashyarmanesh,  \textit{Repeated-root two-dimensional consta\-cyc\-lic codes of
    lenght $2p^s\cdot2^k$}, Finite Fields Appl., \textbf{50} (2018), 122--137.

\bibitem{Rajput}  C. Rajput and M. Bhaintwal,\textit{On the locality of quasi-cyclic codes over finite fields}, Des. Codes Cryptogr., \textit{90} (2022), 759--777.

\bibitem{roy} S. Roy and S. S. Garani, \textit{Two-dimensional algebraic codes for multiple burst error correction},
    IEEE Commun. Letters, \textbf{23}(10) (2019), 1684--1687.

\bibitem{sepasdar} Z. Sepasdar and K. Khashyarmanesh, \textit{Characterizations of some two-dimensional cyclic codes
    correspond  to the ideals of $\mathbb{F}[x,y]/\langle x^s-1, y^{2^k}-1 \rangle$}, Finite Fields Appl., \textbf{41}
    (2016),   97--112.

\bibitem{consta} A. Sharma and S. Rani, \textit{On constacyclic codes over finite fields}, Cryptogr. Commun., \textbf{8} (2016), 617--636.

\bibitem{Sole} M. Shi, S. Li and P. Sol\'e, \textit{$\z_2\z_4$-Additive quasi-cyclic codes}, IEEE Trans. Inform. Theory, \textbf{67}(11) (2021), 7232--7239

\bibitem{van lint} J. H. van Lint, \textit{Repeated-root cyclic codes}, IEEE Trans. Inform. Theory, \textbf{37}(2) (1991), 343--345.
\end{thebibliography}
\end{document}